\documentclass[11pt,tbtags,reqno]{amsart}

\usepackage{amsmath,amsthm,url}

\usepackage{graphicx}
\usepackage{pstricks}
\usepackage{pstricks-add}
\usepackage{pst-node}

\numberwithin{equation}{section}

\let\n\noindent

\DeclareMathOperator{\sng}{sng}
\DeclareMathOperator{\str}{str}
\DeclareMathOperator{\sdet}{sdet}
\DeclareMathOperator{\tr}{tr}
\DeclareMathOperator{\ima}{i}

\DeclareMathOperator*{\Res}{Res}

\newcommand{\ee}[1]{{\,{\rm e}^{\, #1}}}

\newcommand{\trans}{\ensuremath{\mathrm{t}}}

\newcommand{\bmu}{\begin{multline}}
\newcommand{\emu}{\end{multline}}
\newcommand{\beq}{\begin{equation}}
\newcommand{\eeq}{\end{equation}}
\newcommand{\bea}{\begin{eqnarray}}
\newcommand{\eea}{\end{eqnarray}}

\renewcommand{\and}{{\qquad \text{and} \qquad}}

\newtheorem{theorem}{Theorem}
\newtheorem{lemma}[theorem]{Lemma}
\newtheorem{proposition}[theorem]{Proposition}

\newtheorem{corollary}[theorem]{Corollary}

\theoremstyle{remark}
\newtheorem{remark}[theorem]{Remark}

\begin{document}

\title[Supermatrix models]{Supermatrix models, loop equations, and duality}

\author{Patrick Desrosiers}

\address{Instituto Matem\'atica y F\'isica, Universidad de Talca, 2 Norte 685, Talca, Chile.}

\author{Bertrand Eynard}

\address{Institut de Physique Th\'eorique, CEA--Saclay, 91191
 Gif-sur-Yvette cedex, France.}

\email{patrick.desrosiers@inst-mat.utalca.cl}
\email{bertrand.eynard@cea.fr}

\thanks{IPhT--t09/163}

\begin{abstract}
We study integrals over Hermitian supermatrices of arbitrary size $p+q$,
that are parametrized by an external field $X$ and a source $Y$,
of respective size $m+n$ and $p+q$.  We show that these integrals exhibit a simple topological expansion in powers of a formal parameter $\hbar$, which can be identified with
$1/(p-q)$.  The loop equation and the associated spectral curve are also obtained.
 The solutions to the loop equation are given in terms of the symplectic invariants introduced in \cite{EOFg}.
 The symmetry property of the latter objects allows us to prove a duality that relates supermatrix models in which the role of $X$ and $Y$ are interchanged. \end{abstract}



\maketitle

{

\renewcommand{\baselinestretch}{1.0}
\setlength{\parskip}{0ex}
\small
\tableofcontents

}

\newpage

\renewcommand{\baselinestretch}{1.0}
\setlength{\parskip}{1ex}

\normalsize

\section{Introduction}
In this article, we study supermatrix integrals that include an external field.  We first prove that they can be expanded topologically.
We then write their loop equations, and show that they are the same as for usual matrix integrals, and therefore, they have the same solution.
In other words, super-matrix integrals' topological expansion is given by the symplectic invariants of \cite{EOFg}. As a consequence, we prove a duality which generalizes that of \cite{Brezin,Des}.  Namely, we define (notations are explained below)
\beq\label{eqPartFS1}
Z_{(m|n),(p|q)}(X,Y) = \frac{\ee{-\frac{1}{2\hbar}\str Y^2}}{z_{p,q}(\hbar)} \,\, \int_{H(p|q)}\, dM\,\, \ee{-{1\over 2\hbar}\,\str (IM)^2}\,\,  \ee{\frac{1}{\hbar}\str IMY}  {\prod_{i=1}^{m} \sdet(x_i-IM) \over \prod_{i=m+1}^{m+n} \sdet(x_i-M) }
\eeq
and we show that, if $X$ and $Y$ diagonal with complex entries,
\beq \label{eqDuality}
Z_{(m|n),(p|q)}(X,Y) = Z_{(p|q),(m|n)}(Y,X)
\eeq
for all $m, n, p, q \geq 0$.  Note that the matrix integrals $Z_{(m|n),(p|q)}(X,Y)$ and  $Z_{(p|q),(m|n)}(Y,X)$ share the same formal parameter $\hbar$.
 The duality exchanges the size of the matrix with the number of sources,
and it exchanges the external field $Y$ of size $p+q$, with the sources $X$ of size $m+n$.

\bigskip

Let us be more explicit.  The ensemble $H(p|q)$ of hermitian supermatrices of size $p+q$, is the set of matrices of the form:
\beq
M = \left(\begin{array}{ll}A & B \cr  C & D\end{array}\right)
\eeq
where $A$ and $D$ are hermitian matrices of respective size $p\times p$ and $q\times q$, and $B$ and $C=B^\dagger$ are fermionic matrices (entries are Grassmann anti-commuting variables) of respective size $p\times q$ and $q\times p$.
A Hermitian supermatrix can be diagonalized by an element of $U(p|q)$, the supergroup of unitary transformations.
For a short review of the theory of Grassmann algebras and supermatrices, see Appendices A and B, which are based on references \cite{Berezin,Frappat, Efetov}.

Consider $P$, a complex-valued function   depending upon an Hermitian supermatrix $M$.  Its expectation value  with Gaussian measure in $H{(p|q)}$   is defined by
\beq
\left\langle P(M)\right\rangle_{M\in GU{(p|q)}}=\int_{H(p|q)} dM e^{-\frac{1}{2\hbar}\str  (I M)^2}P(I M)\Big/ \int dM e^{-\frac{1}{2\hbar}\str  (I M)^2}
\eeq
where  $\str$ is the supertrace and $I$ is a   supermatrix which ensures the convergence of the integral.  We choose
\beq
I= \mathrm{diag} (\overbrace{1,\dots,1}^p,\overbrace{\ima,\ldots,\ima}^q)
\eeq
so that $I$ is an element of $U(p|q)$, the supergroup of unitary transformations.

Our aim is to study the partition function of Gaussian supermatrix models containing   sources $x_j$ as well as an external field $Y$:
\beq\label{eqPartFS}
Z_{(m|n),(p|q)}(X,Y) = \ee{-\frac{1}{2\hbar}\str Y^2}  \left\langle  \ee{\frac{1}{\hbar}\str MY}  {\prod_{i=1}^{m+n} \sdet(x_i-M)^{\sigma(i)} }   \right\rangle_{M\in GU(p|q)}.
\eeq
In the last equation,   $X$  stands for a Hermitian supermatrix of size $m+n$ with eigenvalues $x_j$ (possibly not all distinct),
 so that $x_i-M$ is understood as the matrix whose element in the $j$th row and $k$th column is equal to $x_i\delta_{jk}-M_{jk}$,
while
\beq
\sigma(i)=\sigma_{m,n}(i)=\begin{cases}
+1&\text{if}\quad  i\in \{1,\ldots, m\},\\
-1&\text{if}\quad  i\in \{m+1,\ldots, m+n\}.
\end{cases}
\eeq


It is worth mentioning that models  involving  hermitian supermatrix have been studied in the past.
For instance,  Itzykson-Zuber and character formulas for $U(n|m)$  have been  obtained  in \cite{Alfaro}.  This had been preceded, in in the beginning
of the 1990s, by a few attempts to generalize the
well known connection between conventional matrix models and quantum
gravity in 2D \cite{Alvarez,Dadda,Yost}. It was soon realized however that  supermatrix models could not provide a discrete version of
 supergravity.  \footnote{Note that other routes have been followed for describing
 supergravity as a matrix models, such as the formal supereigenvalue
 models (see \cite{Plefka} for a review) and matrix model in superspace such as the Marinari-Parisi
 model (see \cite{Verlinde} and references therein).}
 Especially, convincing arguments were given for the equivalence between supermatrix and matrix models when no external fields are involved \cite{Alvarez}.

 Here we indeed prove, in the first section, that we can formally map the partition function of a matrix model to that of supermatrix model.
 It should be understood that this bijection remains true as long as the models are interpreted as linear combination of expectation values
 with respect to the
 Gaussian measure, like in Eq.\ \eqref{EqPartFormal}, and as long as all the matrices' size and entries are considered as parameters.
 When considering the models from a more general perspective, based on algebraic geometry,  the relation between matrix
 and supermatrix problems becomes more subtle.   In particular, we show that supermatrix models possess new critical
behaviors and enjoy more symmetry.  Especially, we will prove the following
duality property given in Eq.\ \eqref{eqDuality}.

\section{Topological expansion}

Here we show that the expectation value of powers of superstraces
have a simple interpretation in terms of ribbon graphs (also called fatgraphs). This naturally leads
to the conclusion that the partition function for a matrix model of
hermitian $(p+q)\times (p+q)$ supermatrices that contains an external field is similar to
the partition function for the usual hermitian matrix model with an external field.   As shown below,
the formal power series expansion in $\hbar^2$ of the supermatrix integrals is of topological nature.
This property will allow us, in section \ref{SectionSpectral}, to exploit the uniqueness of the solution to the loop equations.

\begin{proposition} Let $y$ and $Y$ be a $N\times N$ Hermitian matrix and a $(p+q)\times (p+q)$ Hermitian supermatrix, respectively.
Define the map $\phi$ as \footnote{The function  $\phi$ can be used only when the matrix
$y$ and the supermatrix $Y$ as well as their respective sizes, are considered as parameters.
We cannot  for instance set $N=2$ and $y=\mathrm{diag}(1,1)$ and then apply $\phi$.
In fact, $\phi$  is a homomorphism that maps of the algebra of the symmetric polynomials
in the eigenvalues of $y$ to the subalgebra of the polynomials which symmetric in two set of eigenvalues of
$Y$, $y_1,\ldots,y_p$ and $y_{p+1},\ldots,y_{p+q}$,  that becomes independent of $y_p$ if $y_p=y_{p+q}$.   In general $\phi$ is not invertible.}
\beq  \phi (\tr y^n) =\str Y^n \qquad \forall \, n\geq 0.\eeq Note
in particular that the case $n=0$ corresponds to $\phi (N)=p-q$.
Then
\begin{equation}
 \Big\langle \prod_i \str M^{n_i}\ee{\str MY}\Big\rangle_{M\in GU(p|q)}=
 \phi\left(\Big\langle \prod_i \tr m^{n_i}\ee{\tr my}\Big\rangle_{m\in GU(N)}\right).
\end{equation}
\end{proposition}

 \begin{corollary} Let $g$ be a natural number  and  $p_n=\hbar\str Y^n$, where $Y$ is a
Hermitian supermatrix of size $p+q$.  Let moreover $F=F_{(m|n),(p|q)}(X,Y)$ be the free energy of the supermatrix model whose partition function is
$Z=Z_{(m|n),(p|q)}(X,Y)$; that is, $F=-\ln Z$.  Then, the following formal
power series holds  \beq
F=\sum_{g\geq 0}\hbar^{2g-2} F^{(g)}(p_1,p_2,\ldots).\eeq
\end{corollary}

The latter results are in fact obvious reformulations of Theorem \ref{PropCentral}, which will be proved in the following paragraphs.

\begin{proposition}\label{PropAvMatDif} One has
\begin{multline}\label{EqAvMatDif}
\ee{-\frac{\hbar}{2}\str Y^2} \Big\langle M_{i_1 j_1}\cdots M_{i_n j_n}\ee{\str MY}\Big\rangle_{M\in GU(p|q)}=
\\\left(\sigma(j_1)\frac{\partial}{\partial Y_{j_1 i_1}}+\hbar  Y_{i_1 j_1}\right)\circ
\cdots \circ\left(\sigma(j_n)\frac{\partial}{\partial Y_{j_n i_n}}+\hbar  Y_{i_n j_n}\right)\circ\,\, 1\, .
\end{multline}
\end{proposition}
\begin{proof}By definition we have simply
\begin{equation}
\Big\langle M_{i_1 j_1}\cdots M_{i_n j_n}\ee{\str MY}\Big\rangle = \frac{\int dM e^{-\frac{1}{2\hbar}\str  (I M)^2}\ee{\str IMY}(IM)_{i_1 j_1}\cdots (IM)_{i_n j_n}}{ \int dM e^{-\frac{1}{2\hbar}\str  (I M)^2}}\\
 \end{equation}
But a few manipulations give \beq
\ee{\str IMY}(IM)_{ij}  =\sigma(j)\frac{\partial}{\partial Y_{ji}}
\ee{ \str IMY} \eeq so that
 \begin{multline} \Big\langle M_{i_1 j_1}\cdots M_{i_n j_n}\ee{\str MY}\Big\rangle =\\ \sigma(j_1)\frac{\partial}{\partial Y_{j_1i_1}}\cdots \sigma(j_n)\frac{\partial}{\partial Y_{j_ni_n}}
 \frac{\int dM e^{-\frac{1}{2\hbar}\str  (I M)^2}\ee{\str IMY}}{ \int dM e^{-\frac{1}{2\hbar}\str  (I
 M)^2}}=\\\sigma(j_1)\frac{\partial}{\partial Y_{j_1i_1}}\cdots \sigma(j_n)\frac{\partial}{\partial Y_{j_ni_n}}\Big\langle \ee{\str MY}\Big\rangle
 \end{multline}
Note that the order of the derivatives is important.  We now make
use of the Gaussian integral formula \eqref{eqGaussInt} and obtain:
\beq \Big\langle M_{i_1 j_1}\cdots M_{i_n j_n}\ee{\str
MY}\Big\rangle = \sigma(j_1)\frac{\partial}{\partial
Y_{j_1i_1}}\cdots \sigma(j_n)\frac{\partial}{\partial
Y_{j_ni_n}}\ee{\frac{\hbar}{2}\str Y^2} \eeq Finally, we note that
\beq \ee{-\frac{\hbar}{2}\str Y^2} \frac{\partial}{\partial Y_{ij}}
\ee{\frac{\hbar}{2}\str Y^2}=\frac{\partial}{\partial Y_{ij}}+\hbar
\sigma(i) Y_{ji} \eeq and the proposition follows.
\end{proof}

We now wish to evaluate the expectation value of products of matrix elements by means of a simple generalization of ribbon graphs.  We
define a ribbon graph of order $n$ as a sequence of $n$ vertices graphically
represented as half-edges $\uparrow\, \downarrow$ and ordered from left to right on an horizontal axis.
Each half-edge $\uparrow \,\downarrow$ is labeled by a pair $i\,j$ of
positive integers.  Moreover, a vertex $\uparrow \, \downarrow$  can be connected to at most one other vertex as long as the orientation
of the arrows is respected;
 the connection of two half-edges produces an edge.
Fig.\ref{FigRibbonGraph} gives an example of a ribbon graph of order 6 and labeled by $i_1j_1,\ldots,i_6j_6$.


\begin{figure}[h]\caption{ {\small A ribbon graph with 2 edges and 2 half-edges }}
\centering
 \psset{unit=0.8cm}
\begin{pspicture}(0,0)(9,4.5)
 \psline[linewidth=1pt,linearc=.3,ArrowInside=->](5,1)(5,2)
\psline[linewidth=1pt,linearc=.4,ArrowInside=-<](5.5,1)(5.5,2)
\psline[linewidth=1pt,linearc=.3,ArrowInside=->](8,1)(8,2)
\psline[linewidth=1pt,linearc=.4,ArrowInside=-<](8.5,1)(8.5,2)
\psline[ArrowInside=-<,linewidth=1pt,linearc=.3](2.5,1)(2.5,3)(6.5,3)(6.5,1)
\psline[linewidth=1pt,linearc=.4,ArrowInside=-<](7,1)(7,3.5)(2,3.5)(2,1)
\pscustom[linewidth=6pt,linecolor=white,fillstyle=solid,fillcolor=white]{%
\psline[linewidth=1pt,linearc=.4,border=2pt,ArrowInside=-<](4,1)(4,2.5)(0.5,2.5)(0.5,1)
\psline[linewidth=1pt,linearc=.3,border=2pt,ArrowInside=-<](1,1)(1,2)(3.5,2)(3.5,1)
}
\psline[linewidth=1pt,linearc=.4,ArrowInside=-<](4,1)(4,2.5)(0.5,2.5)(0.5,1)
\psline[linewidth=1pt,linearc=.3,ArrowInside=-<](1,1)(1,2)(3.5,2)(3.5,1)
\pscustom[linewidth=6pt,linecolor=white,fillstyle=solid,fillcolor=white]{%
 \psline[linewidth=1pt,linearc=.3]{->}(5,1)(5,4)
\psline[linewidth=1pt,linearc=.4]{-<}(5.5,4)(5.5,1) }
 \psline[linewidth=1pt,linearc=.3]{->}(5,1)(5,4)
\psline[linewidth=1pt,linearc=.4]{>-}(5.5,4)(5.5,1)
\pscustom[linewidth=6pt,linecolor=white,fillstyle=solid,fillcolor=white]{%
 \psline[linewidth=1pt,linearc=.3](8,1)(8,4)
\psline[linewidth=1pt,linearc=.4](8.5,4)(8.5,1) }
 \psline[linewidth=1pt,linearc=.3 ]{->}(8,1)(8,4)
\psline[linewidth=1pt,linearc=.4 ]{>-}(8.5,4)(8.5,1)
\rput(0.5,0.9){$\bullet$}\rput(1,0.9){$\bullet$}\rput(2,0.9){$\bullet$}\rput(2.5,0.9){$\bullet$}\rput(3.5,0.9){$\bullet$}\rput(4,0.9){$\bullet$}
\rput(8.5,0.9){$\bullet$}\rput(5,0.9){$\bullet$}\rput(8.0,0.9){$\bullet$}\rput(5.5,0.9){$\bullet$}\rput(6.5,0.9){$\bullet$}\rput(7,0.9){$\bullet$}
\rput(0.5,0.5){$i_1$}\rput(1,0.5){$j_1$}\rput(2,0.5){$i_2$}\rput(2.5,0.5){$j_2$}\rput(3.5,0.5){$i_3$}\rput(4,0.5){$j_3$}
\rput(5,0.5){$i_4$}\rput(5.5,0.5){$j_4$}\rput(6.5,0.5){$i_5$}\rput(7,0.5){$j_5$}\rput(8,0.5){$i_6$}\rput(8.5,0.5){$j_6$}
\end{pspicture}
\label{FigRibbonGraph}
\end{figure}
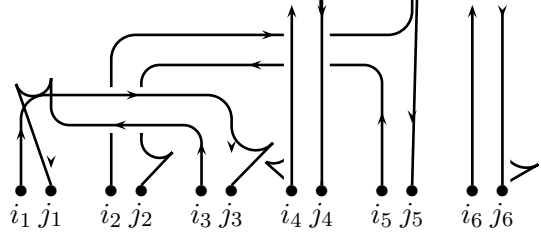

Any collection of independent graphs, $\{G_1,G_2,\ldots\}$, will be written as a sum $G_1+G_2+\ldots$.  The concatenation (or union) of two independent graphs, $G_1$ and $G_2$,   gives another graph $G=G_1G_2$.
For instance, if\\
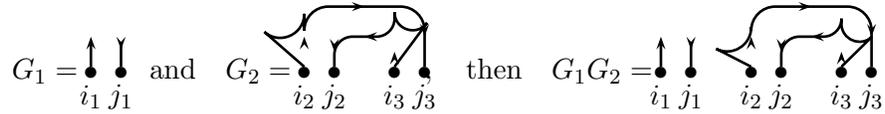
$$
G_1=\,\psset{unit=0.8cm}\begin{pspicture}(0.5,0.8)(1,1.8)
\psline[linewidth=1pt,linearc=.3]{->}(0.5,1)(0.5,1.5)
\psline[linewidth=1pt,linearc=.4]{>-}(1,1.5)(1,1)
\rput(0.5,0.9){$\bullet$}\rput(1,0.9){$\bullet$}
\rput(0.5,0.5){$i_1$}\rput(1,0.5){$j_1$}
\end{pspicture}
\quad\text{and}\quad
G_2=\,
\psset{unit=0.8cm}\begin{pspicture}(0.5,0.8)(2.5,1.8)
\psline[ArrowInside=->,linewidth=1pt,linearc=.4](0.5,1)(0.5,2)(2.5,2)(2.5,1)
\psline[linewidth=1pt,linearc=.3,ArrowInside=->](2,1)(2,1.5)(1,1.5)(1,1)
\rput(0.5,0.9){$\bullet$}\rput(1,0.9){$\bullet$}\rput(2,0.9){$\bullet$}\rput(2.5,0.9){$\bullet$}
\rput(0.5,0.5){$i_2$}\rput(1,0.5){$j_2$}\rput(2,0.5){$i_3$}\rput(2.5,0.5){$j_3$}
\end{pspicture},
\quad \text{then}\quad G_1G_2=\psset{unit=0.8cm}\begin{pspicture}(0.5,0.8)(4,1.8)
\psline[linewidth=1pt,linearc=.3]{->}(0.5,1)(0.5,1.5)
\psline[linewidth=1pt,linearc=.4]{>-}(1,1.5)(1,1)
\rput(0.5,0.9){$\bullet$}\rput(1,0.9){$\bullet$}
\rput(0.5,0.5){$i_1$}\rput(1,0.5){$j_1$}
\psline[ArrowInside=->,linewidth=1pt,linearc=.4](2,1)(2,2)(4,2)(4,1)
\psline[linewidth=1pt,linearc=.3,ArrowInside=->](3.5,1)(3.5,1.5)(2.5,1.5)(2.5,1)
\rput(2,0.9){$\bullet$}\rput(2.5,0.9){$\bullet$}\rput(3.5,0.9){$\bullet$}\rput(4,0.9){$\bullet$}
\rput(2,0.5){$i_2$}\rput(2.5,0.5){$j_2$}\rput(3.5,0.5){$i_3$}\rput(4,0.5){$j_3$}
\end{pspicture}
$$\\

In order to take into account the presence of Grassmann odd variables,
we need to equip the vertices with the following  $\mathbb{Z}_2-$grading: The half-edge labeled by $ij$ has degree $\epsilon_i+\epsilon_j$, where $\epsilon_i$
is defined by
 $\sigma(i)=(-1)^{\epsilon_i}$, that is
\beq \label{EqEpsilon}
\epsilon_i=\epsilon(i)=
\begin{cases}
0&\text{if}\quad  i\in \{1,\ldots, p\},\\
1&\text{if}\quad  i\in \{p+1,\ldots, p+q\}.
\end{cases}
\eeq
The total degree of a graph or a subgraph $G$ is equal to the sum of the degrees of all the half-edges contained in $G$; that is, if $G$ can
 be separated into two independent subgraphs as
 $G=G_1G_2$, then $\deg G=\deg G_1+\deg G_2$ and the connection between half-edges doesn't affect their degree.


\begin{figure}[h]\caption{ {\small Weight function $W$ on the basic constituents (or propagators) of ribbon graphs.} }
\centering
\psset{unit=0.8cm}\begin{pspicture}(-4,0)(9,3)
\rput(-4,1.6){$W$}
\psline[linewidth=1pt,linearc=.3]{->}(-3.5,1)(-3.5,2.5)
\psline[linewidth=1pt,linearc=.3]{-<}(-3,1)(-3,2.5)
\rput(-3.5,1){$\bullet$}\rput(-3,1){$\bullet$}
\rput(-3.5,0.6){$i_1$}\rput(-3,0.6){$j_1$} \rput(-2,1.6){$=\hbar
Y_{i_1j_1}$}
 \rput(0,1.6){$W$}
\psline[ArrowInside=->,linewidth=1pt,linearc=.4](0.5,1)(0.5,2.5)(2.5,2.5)(2.5,1)
\psline[linewidth=1pt,linearc=.3,ArrowInside=->](2,1)(2,2)(1,2)(1,1)
\rput(0.5,1){$\bullet$}\rput(1,1){$\bullet$}\rput(2,1){$\bullet$}\rput(2.5,1){$\bullet$}
\rput(0.5,0.6){$i_1$}\rput(1,0.6){$j_1$}\rput(2,0.6){$i_2$}\rput(2.5,0.6){$j_2$}
\rput(7,1.6){$=\displaystyle \sigma(j_1)\frac{\partial }{\partial
Y_{j_1i_1}}(\hbar Y_{i_2j_2} )=\hbar
\sigma(j_1)\delta_{i_1j_2}\delta_{j_1i_2}$}
\end{pspicture}
 \label{figW1}
\end{figure}

We now introduce a weight function $W$ on graded ribbon graphs whose values belong to a Grassmann algebra over $\mathbb{C}$.
We first set $W(\emptyset)=0$, where $\emptyset$ stands for the empty graph.
Then we fix the weight of a half-edge or an edge as in Fig.\ref{figW1};
it is in correspondence with the matrix operations given in Eq.\eqref{EqAvMatDif}.
Finally, we impose the distributivity and commutation rules of Fig.\ref{figW2} and Fig.\ref{figW3}.
The recursive application of these rules allows us to reduce the evaluation of a whole diagram's weight to a product of weights of edges and half-edges.

The definition of the weight function is such that,
when applying $W$ on a ribbon graph $G$ of order $n$,  one can take any permutation $\pi$ of the half-edges and then multiply the
 weight of the new graph by the appropriate signum.
 In symbols, let $G$ be a graph labeled by $(i_1j_1,\ldots,i_nj_n)$, $\pi$ a permutation of $(1,\ldots,n)$,
 $$\pi(i_1j_1,\ldots,i_nj_n)=(i_{\pi(i_1)}j_{\pi(j_1)},\ldots,i_{\pi(i_n)}j_{\pi(j_n)}),$$ and let $\pi G$ denote the graph obtained
 by permuting the half-edges
 while keeping the links between the edges, then
  \beq W(G)=\sng (\pi)W(\pi G), \eeq  where $\sng (\pi)$ can be evaluated thanks to
  \beq Y_{i_1j_1}\cdots Y_{i_nj_n}=\sng (\pi)Y_{i_{\pi(1)}j_{\pi(1)}}\cdots
   Y_{i_{\pi(n)}j_{\pi(n)}}.\eeq
Consider for instance the graph
given in Fig.\ref{FigRibbonGraph}.  Then its weight is equal to
$$(-1)^{(\epsilon_{i_2}+\epsilon_{j_2})(\epsilon_{i_3}+\epsilon_{j_3})}(-1)^{(\epsilon_{i_4}+\epsilon_{j_4})(\epsilon_{i_5}+\epsilon_{j_5})}(-1)^{\epsilon_{j_1}+\epsilon_{j_2}}
\,\delta_{i_1j_3}\delta_{j_1i_3}\delta_{i_2j_5}\delta_{j_2i_5}\,\hbar^4\ Y_{i_4j_4}Y_{i_6j_6},$$
In the case where $p=2,q=2$ and $$(i_1,j_1,i_2,j_2,i_3,j_3,i_4,j_4,i_5,j_5,i_6,j_6)=(1,3,2,4,3,1,4,4,1,3,2,3),$$
the weight simplifies to  $-\hbar^4 Y_{44}Y_{23}$.

\begin{figure}[h] \caption{ {\small Weight function $W$ on independent subgraphs $G_1$ and $G_2$.
The signum depends on the degree of the components: $\epsilon=\deg G_1\cdot\deg G_2$ and  $\bar{\epsilon}=(\epsilon_i +\epsilon_j)\cdot\deg G_1$.  }}.
\centering
$$
\psset{unit=0.8cm}\begin{pspicture}(0,0)(10,1)
\rput(0.4,0.5){$W$}
\pscurve[linewidth=0.5pt](0.8,0)(.7,0.5)(.8,1)
 \psline[linewidth=.5pt]{-}(2,0)(1,0)(1,1)(2,1)(2,0)\rput(1.5,0.5){$G_1$}
\psline[linewidth=.5pt]{-}(3.2,0)(2.2,0)(2.2,1)(3.2,1)(3.2,0)\rput(2.7,0.5){$G_2$}
\pscurve[linewidth=0.5pt](3.4,0)(3.5,0.5)(3.4,1)
\rput(4,0.5){$=$}
\rput(4.6,0.5){$W$}
\pscurve[linewidth=0.5pt](5,0)(4.9,0.5)(5,1)
 \psline[linewidth=.5pt]{-}(6.2,0)(5.2,0)(5.2,1)(6.2,1)(6.2,0)\rput(5.7,0.5){$G_1$}
 \pscurve[linewidth=0.5pt](6.4,0)(6.5,0.5)(6.4,1)
 \rput(6.8,0.5){$W$}
\pscurve[linewidth=0.5pt](7.2,0)(7.1,0.5)(7.2,1)
 \psline[linewidth=.5pt]{-}(8.4,0)(7.4,0)(7.4,1)(8.4,1)(8.4,0)\rput(7.9,0.5){$G_2$}
 \pscurve[linewidth=0.5pt](8.6,0)(8.7,0.5)(8.6,1)
 \rput(9.0,0.5){$=$}
 \rput(9.7,0.5){$(-1)^{\epsilon}$}
\end{pspicture}
\begin{pspicture}(9.8,0)(15,1)
 \rput(10.4,0.5){$W$}
\pscurve[linewidth=0.5pt](10.8,0)(10.7,0.5)(10.8,1)
 \psline[linewidth=.5pt]{-}(12,0)(11,0)(11,1)(12,1)(12,0)\rput(11.5,0.5){$G_2$}
 \pscurve[linewidth=0.5pt](12.2,0)(12.3,0.5)(12.2,1)
 \rput(12.6,0.5){$W$}
\pscurve[linewidth=0.5pt](13,0)(12.9,0.5)(13,1)
 \psline[linewidth=.5pt]{-}(14.2,0)(13.2,0)(13.2,1)(14.2,1)(14.2,0)\rput(13.7,0.5){$G_1$}
 \pscurve[linewidth=0.5pt](14.4,0)(14.5,0.5)(14.4,1)
\end{pspicture}
$$

$$\psset{unit=0.8cm}\begin{pspicture}(2,0)(9.5,3.5)
\rput(2.4,2){$W$}
\pscurve[linewidth=0.5pt](2.8,3)(2.7,2)(2.8,1)\pscurve[linewidth=0.5pt](6.7,3)(6.8,2)(6.7,1)
 \psline[linewidth=.5pt]{-}(5.5,1)(6.5,1)(6.5,3)(5.5,3)(5.5,1)\rput(6,2){$G_2$}
 \psline[linewidth=.5pt]{-}(4,1)(5,1)(5,2)(4,2)(4,1)\rput(4.5,1.5){$G_1$}
\pscustom[linewidth=8pt,linecolor=white,fillstyle=solid,fillcolor=white]{%
\psline[linewidth=1pt,linearc=.4,ArrowInside=->](3,1)(3,3)(5.5,3)
\psline[linewidth=1pt,linearc=.3,ArrowInside=->](5.5,2.5)(3.5,2.5)(3.5,1)}
\psline[linewidth=1pt,linearc=.4,ArrowInside=->](3,1)(3,3)(5.5,3)
\psline[linewidth=1pt,linearc=.3,ArrowInside=->](5.5,2.5)(3.5,2.5)(3.5,1)
\rput(3,1){$\bullet$}\rput(3.5,1){$\bullet$}
\rput(3,0.6){$i$}\rput(3.5,0.6){$j$}
 \rput(8,2){$\displaystyle =(-1)^{\bar{\epsilon}}\,W$}
 \end{pspicture}
 \psset{unit=0.8cm}\begin{pspicture}(9,0)(15,3.5)
 \pscurve[linewidth=0.5pt](8.9,3)(8.8,2)(8.9,1)
\psline[linewidth=.5pt]{-}(9,1.5)(10,1.5)(10,2.5)(9,2.5)(9,1.5)\rput(9.5,2){$G_1$}
\pscurve[linewidth=0.5pt](10.1,3)(10.2,2)(10.1,1)
\rput(10.45,2){$W$}
\pscurve[linewidth=0.5pt](10.8,3)(10.7,2)(10.8,1)
 \psline[linewidth=.5pt]{-}(12.5,1)(13.5,1)(13.5,3)(12.5,3)(12.5,1)\rput(13,2){$G_2$}
 \pscustom[linewidth=8pt,linecolor=white,fillstyle=solid,fillcolor=white]{%
\psline[linewidth=1pt,linearc=.4,ArrowInside=->](11,1)(11,3)(12.5,3)
\psline[linewidth=1pt,linearc=.3,ArrowInside=->](12.5,2.5)(11.5,2.5)(11.5,1)}
\psline[linewidth=1pt,linearc=.4,ArrowInside=->](11,1)(11,3)(12.5,3)
\psline[linewidth=1pt,linearc=.3,ArrowInside=->](12.5,2.5)(11.5,2.5)(11.5,1)
\pscurve[linewidth=0.5pt](13.7,3)(13.8,2)(13.7,1)
\rput(11,1){$\bullet$}\rput(11.5,1){$\bullet$}
\rput(11,0.6){$i$}\rput(11.5,0.6){$j$}
\end{pspicture}
$$\label{figW2}
\end{figure}



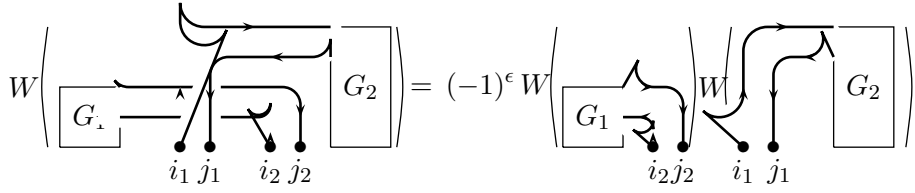
\begin{figure}[h] \caption{ {\small Weight $W$ associated to a crossing of ribbons.
The signum depends on the degree of each ribbon, that
is, ${\epsilon}=(\epsilon_{i_1} +\epsilon_{j_1})(\epsilon_{i_2} +\epsilon_{j_2})$. } }
\centering
\psset{unit=0.8cm}
\begin{pspicture}(1,0)(9,4)
\rput(0.4,2){$W$}
\pscurve[linewidth=0.5pt](0.8,3)(.7,2)(.8,1)\pscurve[linewidth=0.5pt](6.6,3)(6.7,2)(6.6,1)
 \psline[linewidth=.5pt]{-}(5.5,1)(6.5,1)(6.5,3)(5.5,3)(5.5,1)\rput(6,2){$G_2$}
 \psline[linewidth=.5pt]{-}(1,1)(2,1)(2,2)(1,2)(1,1)\rput(1.5,1.5){$G_1$}
 \pscustom[linewidth=8pt,linecolor=white,fillstyle=solid,fillcolor=white]{%
\psline[linewidth=1pt,linearc=.4,ArrowInside=->](4.5,1)(4.5,1.5)(2,1.5)
\psline[linewidth=1pt,linearc=.3,ArrowInside=->](2,2)(5,2)(5,1)}
\psline[linewidth=1pt,linearc=.2,ArrowInside=->](4.5,1)(4.5,1.5)(2,1.5)
\psline[linewidth=1pt,linearc=.3,ArrowInside=->](2,2)(5,2)(5,1)
\pscustom[linewidth=8pt,linecolor=white,fillstyle=solid,fillcolor=white]{%
\psline[linewidth=1pt,linearc=.4,ArrowInside=->](3,1)(3,3)(5.5,3)
\psline[linewidth=1pt,linearc=.3,ArrowInside=->](5.5,2.5)(3.5,2.5)(3.5,1)}
\psline[linewidth=1pt,linearc=.4,ArrowInside=->](3,1)(3,3)(5.5,3)
\psline[linewidth=1pt,linearc=.3,ArrowInside=->](5.5,2.5)(3.5,2.5)(3.5,1)
\rput(3,1){$\bullet$}\rput(3.5,1){$\bullet$}\rput(4.5,1){$\bullet$}\rput(5,1){$\bullet$}
\rput(3,0.6){$i_1$}\rput(3.5,0.6){$j_1$}\rput(4.5,0.6){$i_2$}\rput(5,0.6){$j_2$}
 \rput(8,2){$\displaystyle =\,(-1)^{{\epsilon}}\,W$}
\end{pspicture}
 \psset{unit=0.8cm}
\begin{pspicture}(8.8,0)(15,4)
\psline[linewidth=.5pt]{-}(9,1)(10,1)(10,2)(9,2)(9,1)\rput(9.5,1.5){$G_1$}
 \pscustom[linewidth=8pt,linecolor=white,fillstyle=solid,fillcolor=white]{%
\psline[linewidth=1pt,linearc=.4,ArrowInside=->]{->}(10.5,1)(10.5,1.5)(10,1.5)
\psline[linewidth=1pt,linearc=.3,ArrowInside=->]{->}(10,2)(11,2)(11,1)}
\psline[linewidth=1pt,linearc=.2,ArrowInside=->](10.5,1)(10.5,1.5)(10,1.5)
\psline[linewidth=1pt,linearc=.3,ArrowInside=->](10,2)(11,2)(11,1)
 \pscurve[linewidth=0.5pt](8.9,3)(8.8,2)(8.9,1)
\pscurve[linewidth=0.5pt](11.1,3)(11.2,2)(11.1,1) \rput(11.5,2){$W$}
\pscurve[linewidth=0.5pt](11.8,3)(11.7,2)(11.8,1)
 \psline[linewidth=.5pt]{-}(13.5,1)(14.5,1)(14.5,3)(13.5,3)(13.5,1)\rput(14,2){$G_2$}
 \pscustom[linewidth=8pt,linecolor=white,fillstyle=solid,fillcolor=white]{%
\psline[linewidth=1pt,linearc=.4,ArrowInside=->](12,1)(12,3)(13.5,3)
\psline[linewidth=1pt,linearc=.3,ArrowInside=->](13.5,2.5)(12.5,2.5)(12.5,1)}
\psline[linewidth=1pt,linearc=.4,ArrowInside=->](12,1)(12,3)(13.5,3)
\psline[linewidth=1pt,linearc=.3,ArrowInside=->](13.5,2.5)(12.5,2.5)(12.5,1)
\pscurve[linewidth=0.5pt](14.7,3)(14.8,2)(14.7,1)
\rput(10.5,1){$\bullet$}\rput(11,1){$\bullet$}\rput(12,1){$\bullet$}\rput(12.5,1){$\bullet$}
\rput(10.6,0.6){$i_2$}\rput(11,0.6){$j_2$}\rput(12,0.6){$i_1$}\rput(12.6,0.6){$j_1$}
\end{pspicture}
\label{figW3}
\end{figure}


\begin{proposition}\label{PropBijection} Let $\Sigma G$ denote the sum of all possible ribbon graphs of degree $n$ labeled by $i_1j_i,\ldots,i_nj_n$.  In order words,
let $\Sigma G$ be equal to
\beq\label{eqGraph}
\psset{unit=0.8cm}
\begin{pspicture}(0,0.5)(5,1.8)
\psline[linewidth=1pt,linearc=.3]{->}(0.5,1)(0.5,1.8)
\psline[linewidth=1pt,linearc=.4]{>-}(1,1.8)(1,1)
\psline[linewidth=1pt,linearc=.3]{->}(2,1)(2,1.8)
\psline[linewidth=1pt,linearc=.4]{>-}(2.5,1.8)(2.5,1)
 \psline[linewidth=1pt,linearc=.3]{->}(4,1)(4,1.8)
\psline[linewidth=1pt,linearc=.4]{>-}(4.5,1.8)(4.5,1)
\rput(0.5,0.9){$\bullet$}\rput(1,0.9){$\bullet$}\rput(2,0.9){$\bullet$}\rput(2.5,0.9){$\bullet$}\rput(3.5,0.9){$\dots$}\rput(4,0.9){$\bullet$}
\rput(4.5,0.9){$\bullet$}
\rput(0.5,0.5){$i_1$}\rput(1,0.5){$j_1$}\rput(2,0.5){$i_2$}\rput(2.5,0.5){$j_2$}\rput(4,0.5){$i_n$}\rput(4.5,0.5){$j_n$}
\end{pspicture}
\eeq
plus all distinct graphs obtained from the latter by connecting at least 2 half-edges and at most $n$ or $n-1$ of them,
depending on wether $n$ is even or odd, respectively.
Then
$$\ee{-\frac{\hbar}{2}\str Y^2} \Big\langle M_{i_1 j_1}\cdots M_{i_n j_n}\ee{\str MY}\Big\rangle_{M\in GU(p|q)}=W(\Sigma G)$$
 \end{proposition}
\begin{proof}Let $G$ be an arbitrary ribbon graph and let $1=\mathrm{id}$ be the left identity operator on graphs, which means $1\circ G$=$G$.  We introduce $A_{ij}$, a
noncommutative operator acting on $G$ by creating a half-edge, labeled  by $ij$, to the left of $G$.
Obviously, the graph $O_n$ given in Eq.\eqref{eqGraph} can be written as
$ A_{i_1j_1}\ldots A_{i_nj_n}\emptyset.$ Let also $B_{ij,k\ell}$ be a noncommutative operator
acting on $G$ by connecting the half-edge labeled by $ij$ to that labeled by $k\ell$, the first half-edge being located to the left of the second.
We consider $B_{ij}$ as a nilpotent operator, which means $B_{ij}^2G=\emptyset$, and set $B_{ij,k\ell}G=\emptyset$ if $G$ doesn't contain a half-edge labeled by $ij$ (or $k\ell$) or if all half-edges
$ij$ (or $k\ell$) in $G$ are already connected.
Thus, every ribbon graph of degree $n$ can be uniquely written as a
polynomial in operators $B_{i,j}$ acting from the left on $O_n=A_{i_1j_1}\ldots A_{i_nj_n}\emptyset$.  In particular, the sum $\Sigma G$ of all graphs obtained from $O_n$ by connecting the half-edges in all possible ways, is given by
\begin{multline}\Sigma G= (1+B_{i_1j_1,i_2j_2})(1+B_{i_1j_1,i_3j_3}) \cdots (1+B_{i_1j_1,i_nj_n})\\
 \circ (1+B_{i_2j_2,i_3j_3})\cdots (1+B_{i_2j_3,i_nj_n})\cdots (1+B_{i_{n-1}j_{n-1},i_{n}j_{n}}) O_n.\end{multline}
>From the properties of the connection operators $B$, we can simplify the latter expression as follows:
\begin{multline}\label{EqPropSigmaG}\Sigma G= (1+B_{i_1j_1,i_2j_2}+B_{i_1j_1,i_3j_3} \ldots B_{i_1j_1,i_nj_n})\\
 \circ (1+B_{i_2j_2,i_3j_3}+\ldots+B_{i_2j_3,i_nj_n})\cdots (1+B_{i_{n-1}j_{n-1},i_{n}j_{n}}) O_n.\end{multline}

Now we consider the action weight function on the graph $\Sigma G$.  We first note that
$$ W(\sum_{\ell>k}B_{i_kj_k,i_lj_l}O_n)=W(A_{i_1j_1}\ldots A_{i_{k-1}j_{k-1}})\cdot W( \sum_{\ell>k}B_{i_kj_k,i_lj_l} A_{i_1j_1}\ldots A_{i_nj_n}\emptyset).$$  The latter formula and Eq.\eqref{EqPropSigmaG} then imply
\begin{multline}\label{EqPropSigmaG2}W(\Sigma G)= W\Big((1+B_{i_1j_1,i_2j_2}+B_{i_1j_1,i_3j_3} \ldots B_{i_1j_1,i_nj_n})A_{i_1j_1}\\
 \circ (1+B_{i_2j_2,i_3j_3}+\ldots+B_{i_2j_3,i_nj_n})A_{i_2j_2}\cdots (1+B_{i_{n-1}j_{n-1},i_{n}j_{n}})A_{i_nj_n}\emptyset\Big),\end{multline}
 Using the fact that $\sum_{\ell> k}B_{i_kj_k,i_\ell j_\ell}A_{i_kj_k}G=\sum_{\ell\neq k}B_{i_kj_k,i_\ell j_\ell}A_{i_kj_k}G$ when $G$ doesn't contain half-edges labeled by $i_\ell j_\ell$ such that $\ell < k$, we get
\beq W(\Sigma G)=W(D_{i_1j_1} \cdots D_{i_nj_n}\emptyset), \quad D_{i_kj_k}=A_{i_kj_k}+\sum_{\ell\neq k}B_{i_kj_k,i_\ell j_\ell}A_{i_kj_k},\eeq which corresponds the successive concatenation of a half-edge followed by its connection (or not) to all possible half-edges contained in the graph  on its right.   By making use of the commutation rules of Fig.\ref{figW2} and \ref{figW3}, one easily verifies that
$$ W(\Sigma G)=W(\cdots D_{i_kj_k} D_{i_\ell j_\ell}\cdots\emptyset)=(-1)^{(\epsilon_{i_k} +\epsilon_{j_k})(\epsilon_{i_\ell} +\epsilon_{j_\ell})}
W(\cdots D_{i_\ell j_\ell} D_{i_kj_k}\cdots\emptyset)$$ and $$ W(D_{i_kj_k}D_{i_{k+1}j_{k+1}} \cdots D_{i_nj_n}\emptyset)=\left(\hbar Y_{i_kj_k}+\sigma(j_k)\frac{\partial}{\partial Y_{j_ki_k}}\right)\circ W(D_{i_{k+1}j_{k+1}} \cdots D_{i_nj_n}\emptyset).$$
Thus by induction, $$W(\Sigma G)= \left(\hbar Y_{i_1j_1}+\sigma(j_1)\frac{\partial}{\partial Y_{j_1i_1}}\right)\circ\cdots\circ \left(\hbar Y_{i_nj_n}+\sigma(j_n)\frac{\partial}{\partial Y_{j_ni_n}}\right)\circ 1,$$and the proposition follows from the comparison of the latter equation and Lemma \ref{PropAvMatDif}.\end{proof}

An immediate consequence of the latter result is a graded version of the Wick
formula, which, in the context of matrix models, allows to express
the correlation of  $2k$ matrix elements in terms of a sum of
monomials involving $k$ correlations of 2 elements.

In order to present the formula in a compact form, we need to introduce some more notations.
Let $a=(a_1,\ldots,a_k)$ and $b=(b_1,\ldots,b_k)$ be two disjoint increasing sequences of integers in $\{1,\ldots,2k\}$.  By
$\mathcal{P}_{a,b}=\{(i_{a_1}j_{a_1},i_{b_1}j_{b_1}),\ldots,(i_{a_k}j_{a_k},i_{b_k}j_{b_k})\}$  we denote a set of pairings of the indices ${i_1 j_1},\ldots, i_{2k}
j_{2k}$.  Let $\pi_{a,b}\in S_{2k}$  be the permutation such that
$$(a_1,b_1,\ldots,a_k, b_k) =\pi_{a,b}(1,2,\ldots, 2k-1, 2k).$$
Define the signum of $\pi_{a,b}$ by
$$ M_{i_1 j_1}\cdots M_{i_{2k} j_{2k}}= \mathrm{sgn}(\pi_{a,b})M_{i_{a_1}j_{a_1}}M_{i_{b_1}j_{b_1}}\cdots
M_{i_{a_k}j_{a_k}}M_{i_{b_k}j_{b_k}}.$$

\begin{corollary} Let $\Big\langle M_{i_1 j_1}\cdots M_{i_{n} j_{n}} \Big\rangle$ denote the expectation
value given by the left-hand-side  Eq.\eqref{EqAvMatDif} with $Y=0$.   Then,
$$\displaystyle \Big\langle M_{i_k j_k}M_{i_{\ell}j_{n\ell}} \Big\rangle=\hbar\sigma(j_k)\delta_{i_k,j_\ell}\delta_{j_j,i_\ell}.$$
Moreover, if $k$ is a nonnegative integer,
$$ \displaystyle  \Big\langle M_{i_1 j_1}\cdots M_{i_{2k+1} j_{2k+1}} \Big\rangle=0 \displaystyle $$
\beq \Big\langle M_{i_1 j_1}\cdots M_{i_nj_{n}} \Big\rangle=\sum \mathrm{sgn}(\pi_{a,b})
 \Big\langle
M_{i_{a_1}j_{a_1}}M_{i_{b_1}j_{b_1}}\Big\rangle \cdots \Big\langle
M_{i_{a_k}j_{a_k}}M_{i_{b_k}j_{b_k}}\Big\rangle
, \eeq
where the sum runs over all distinct sets of pairings $\mathcal{P}_{a,b}$ of ${i_1 j_1},\ldots, i_{2k}j_{2k}$.
\end{corollary}

We are ready to turn our attention to the expectation value of products of supertraces.  From Proposition \ref{PropBijection} and the definition of the supertrace, we know that
\beq\label{EqFormulaExpectation}\ee{-\frac{\hbar}{2}\str Y^2} \Big\langle \str M^n \ee{\str MY}\Big\rangle_{M\in GU(p|q)}=\sum_{1\leq i_1,\ldots, i_n\leq p+q} \sigma(i_1) \, W(\Sigma T),
\eeq
where $\Sigma T$ denotes the sum of all distinct graphs that we can get by connecting by pairs the half-edges of the following graph:
{ \beq\label{eqGraphT}
\psset{unit=0.85cm}
\begin{pspicture}(0,0.5)(6,1.8)
\psline[linewidth=1pt,linearc=.3]{->}(0.5,1)(0.5,1.8)
\psline[linewidth=1pt,linearc=.4]{>-}(1,1.8)(1,1)
\psline[linewidth=1pt,linearc=.3]{->}(2,1)(2,1.8)
\psline[linewidth=1pt,linearc=.4]{>-}(2.5,1.8)(2.5,1)
 \psline[linewidth=1pt,linearc=.3]{->}(4,1)(4,1.8)
\psline[linewidth=1pt,linearc=.4]{>-}(4.5,1.8)(4.5,1)
 \psline[linewidth=1pt,linearc=.3]{->}(5.5,1)(5.5,1.8)
\psline[linewidth=1pt,linearc=.4]{>-}(6,1.8)(6,1)
 \psline[linewidth=1pt,linearc=.3,ArrowInside=->,linestyle=dotted](4.5,1)(5.5,1)
\psline[linewidth=1pt,linearc=.3,ArrowInside=->,linestyle=dotted](6,1)(6,0.4)(0.5,0.4)(0.5,1)
 \psline[linewidth=1pt,linearc=.3,ArrowInside=->,linestyle=dotted](1,1)(2,1)
  \rput(3.3,1){$\ldots$}
\rput(0.5,1){$\bullet$}\rput(1,1){$\bullet$}\rput(2,1){$\bullet$}\rput(2.5,1){$\bullet$}\rput(4,1){$\bullet$}\rput(4.5,1){$\bullet$}
\rput(6,1){$\bullet$}\rput(5.5,1){$\bullet$}
\rput(0.3,1.3){$i_1$}\rput(1.3,1.3){$i_2$}\rput(2,0.7){$i_2$}\rput(2.5,0.7){$i_3$}\rput(4,0.7){$i_{n-1}$}\rput(4.6,0.7){$i_n$}
\rput(5.3,1.3){$i_n$}\rput(6.2,1.3){$i_1$}
\end{pspicture}
\eeq}
Note that doted arrows
$
\quad \psset{unit=0.8cm}\begin{pspicture}(0,.7)(1,1)\psline[linewidth=1pt,linearc=.3,ArrowInside=->,linestyle=dotted](0,1)(1,1)
\rput(0,1){$\bullet$}\rput(1,1){$\bullet$}
\rput(0,.7){$i_k$}\rput(1,.7){$i_k$}\end{pspicture}\quad$
have been used in order to ease the identification of components that share a same index.  We set $W\left(\;\psset{unit=0.8cm}\begin{pspicture}(0,.7)(1,1)\psline[linewidth=1pt,linearc=.3,ArrowInside=->,linestyle=dotted](0,1)(1,1)
\rput(0,1){$\bullet$}\rput(1,1){$\bullet$}
\rput(0,.7){$i_k$}\rput(1,.7){$i_k$}\end{pspicture}\;\right)=1$.
A ribbon graph or a subgraph $T$ is said to be of trace-type if it is labeled like the graph in Eq.\eqref{eqGraphT},
which means that its indices follow the pattern $$i_ki_{k+1}, \quad i_{k+1}i_{k+2}, \quad i_{k+2}i_{k+3},\quad \ldots,\quad i_{k+\ell}i_k.$$
The degree of a graph of trace-type depends only on the degree of the first  and last labels, for the intermediate indices are always repeated.
 In order words, if $T$ is of trace-type,
  \begin{equation}\label{EqDegTrace}
\deg \left(\quad \psset{unit=0.85cm}\begin{pspicture}(2.5,1.2)(3.5,1.8)
 \psline[linewidth=.5pt]{-}(2.5,1)(2.5,1.8)(3.5,1.8)(3.5,1)(2.5,1)\rput(3,1.4){$T$}
\rput(2.5,1){$\bullet$}\rput(3.5,1){$\bullet$}
\rput(2.5,0.7){$i_k$}\rput(3.5,0.7){$\;i_{\ell}$}
\end{pspicture}\quad\right)= \epsilon_{i_k}+\epsilon_{i_\ell}
\end{equation}\\
Notice that we adopt the following convention: if $T_j$ is an empty subgraph of a trace-type graph,  then
\begin{equation}
 \psset{unit=0.85cm}\begin{pspicture}(2.5,1.2)(3.5,1.8)
 \psline[linewidth=.5pt]{-}(2.5,1)(2.5,1.8)(3.5,1.8)(3.5,1)(2.5,1)\rput(3,1.4){$T_j$}
\rput(2.5,1){$\bullet$}\rput(3.5,1){$\bullet$}
\rput(2.5,0.7){$i_k$}\rput(3.5,0.7){$\;i_{k}$}
\end{pspicture}\quad = \quad \psset{unit=0.85cm} \begin{pspicture}(2.5,1.2)(3.5,1.8)
 \psline[linewidth=1pt,linearc=.3,ArrowInside=->,linestyle=dotted](2.5,1) (3.5,1)
\rput(2.5,1){$\bullet$}\rput(3.5,1){$\bullet$}
\rput(2.5,0.7){$i_k$}\rput(3.5,0.7){$\;i_{k}$}
\end{pspicture}
\end{equation}\\
so that the weight of $T_j$ is equal to one.

 In order to calculate the contribution of graphs of trace type,  we need to establish a few rules that allow to decompose a graph in its independent parts.
 From the definition of the weight $W$ and Eq.\eqref{EqDegTrace}, we see that if $T_1$ and $T_2$ stand for 2 subgraphs of trace type (independent or not), then
 \begin{equation}\label{EqCommuteT}
 W\left(\quad \psset{unit=0.85cm}\begin{pspicture}(2.5,1.2)(5,1.8)
 \psline[linewidth=.5pt]{-}(2.5,1)(2.5,1.8)(3.5,1.8)(3.5,1)(2.5,1)\rput(3,1.4){$T_1$}
  \psline[linewidth=1pt,linearc=.3,linestyle=dotted]{->}(3.5,1) (3.9,1)
  \psline[linewidth=.5pt]{-}(4,1)(4,1.8)(5,1.8)(5,1)(4,1)\rput(4.5,1.4){$T_2$}
\rput(2.5,1){$\bullet$}\rput(3.5,1){$\bullet$}
\rput(2.5,0.7){$i_k$}\rput(3.5,0.7){$\;i_{\ell}$}
\rput(4,1){$\bullet$}\rput(5,1){$\bullet$}
\rput(4,0.7){$i_\ell$}\rput(5,0.7){$\;i_{m}$}
\end{pspicture}\quad\right)  =(-1)^{(\epsilon_k+\epsilon_\ell)(\epsilon_\ell+\epsilon_m)}
 W\left(\quad
 \psset{unit=0.85cm}\begin{pspicture}(2.5,1.2)(5,1.8)
  \psline[linewidth=1pt,linearc=.3,ArrowInside=->,linestyle=dotted](5,1)(5,0.4)(2.5,0.4)(2.5,1)
 \psline[linewidth=.5pt]{-}(2.5,1)(2.5,1.8)(3.5,1.8)(3.5,1)(2.5,1)\rput(3,1.4){$T_2$}
  \psline[linewidth=.5pt]{-}(4,1)(4,1.8)(5,1.8)(5,1)(4,1)\rput(4.5,1.4){$T_1$}
\rput(2.5,1){$\bullet$}\rput(3.5,1){$\bullet$}
\rput(2.3,1.2){$i_\ell$}\rput(3.5,0.7){$\;i_{m}$}
\rput(4,1){$\bullet$}\rput(5,1){$\bullet$}
\rput(4,0.7){$i_k$}\rput(5.2,1.2){$\;i_{\ell}$}
\end{pspicture}\quad\right)
\end{equation}\\
Moreover,
{\small
\begin{equation}\label{EqNonZeroW}
W\left(\quad \psset{unit=0.85cm}\begin{pspicture}(2.5,1.5)(5.5,2.5)
 \psline[linewidth=.5pt]{-}(3.5,1)(3.5,1.8)(4.5,1.8)(4.5,1)(3.5,1)\rput(4,1.4){$T_2$}
   \psline[linewidth=1pt,linearc=.3,linestyle=dotted]{->}(3,1) (3.4,1)
     \psline[linewidth=1pt,linearc=.3,linestyle=dotted]{->}(4.5,1) (4.9,1)
\psline[linewidth=1pt,linearc=.4,ArrowInside=-<](5.5,1)(5.5,2.5)(2.5,2.5)(2.5,1)
\psline[linewidth=1pt,linearc=.3,ArrowInside=-<](3,1)(3,2)(5,2)(5,1)
\rput(2.5,1){$\bullet$}\rput(3,1){$\bullet$}\rput(3.5,1){$\bullet$}
\rput(4.5,1){$\bullet$}\rput(5,1){$\bullet$}\rput(5.5,1){$\bullet$}
\rput(2.5,0.7){$i_k$}\rput(3,0.7){$\;i_{k+1}$}\rput(3.6,0.7){$\;\;\;i_{k+1}$}\rput(4.5,0.7){$i_{\ell}$}
\rput(5,0.7){$i_{\ell}$}\rput(5.5,0.7){$\;i_{\ell+1}$}
\end{pspicture}\quad\right)\neq 0\quad \Longrightarrow \quad i_\ell=i_{k+1},\quad i_{\ell+1}=i_k;
\end{equation}
}

\n This in turn implies that the previous subgraph has a nonzero contribution only if both $T_2$ and the edge linking $i_ki_{k+1}$ and $i_{\ell}i_{\ell+1}$
has a degree equal to zero.  In order words, these subgraphs contribute only if they are bosonic,
so that we can commute all components of that type when evaluating the weight.  By making use of equations \eqref{EqDegTrace},\eqref{EqCommuteT},
\eqref{EqNonZeroW} and the fact that
$(-1)^{(\epsilon_k+\epsilon_\ell)(\epsilon_k+\epsilon_\ell)}=(-1)^{(\epsilon_k+\epsilon_\ell)}$, one can establish the following.

\begin{lemma} \label{LemmaTraceGraph} Let $T_1, \ldots, T_5$ be  trace-type ribbon graphs  that may be dependent or not. Then the equations given in Figures
\ref{FigComTrace}, \ref{FigEffectEdgeTrace}, and  \ref{FigCrossingEdgesTrace} hold.  \end{lemma}

\begin{figure}[h]\caption{ {\small Commutation of trace type ribbon graphs}}
\centering
{
 \begin{equation*}\label{EqCommuteT}
 \sigma(i_1)\,W\left(\quad\;\psset{unit=0.8cm}\begin{pspicture}(2.5,1.2)(5,1.8)
 \psline[linewidth=.5pt ]{-}(2.5,1)(2.5,1.8)(3.5,1.8)(3.5,1)(2.5,1)\rput(3,1.4){$T_1$}
  \psline[linewidth=1pt,linearc=.3,linestyle=dotted]{->}(3.5,1) (3.9,1)
  \psline[linewidth=.5pt ]{-}(4,1)(4,1.8)(5,1.8)(5,1)(4,1)\rput(4.5,1.4){$T_2$}
\rput(2.5,1){$\bullet$}\rput(3.5,1){$\bullet$}
\rput(2.5,0.7){$i_1$}\rput(3.5,0.7){$\;i_{\ell}$}
\rput(4,1){$\bullet$}\rput(5,1){$\bullet$}
\rput(4,0.7){$\;i_{\ell}$}\rput(5,0.7){$\;i_{1}$}
\end{pspicture}\quad\;\right)  = \sigma(i_\ell)\,
 W\left(\quad\psset{unit=0.8cm}\begin{pspicture}(2.5,1.2)(5,1.8)
 \psline[linewidth=.5pt ]{-}(2.5,1)(2.5,1.8)(3.5,1.8)(3.5,1)(2.5,1)\rput(3,1.4){$T_2$}
  \psline[linewidth=1pt,linearc=.3,linestyle=dotted]{->}(3.5,1) (3.9,1)
  \psline[linewidth=.5pt ]{-}(4,1)(4,1.8)(5,1.8)(5,1)(4,1)\rput(4.5,1.4){$T_1$}
\rput(2.5,1){$\bullet$}\rput(3.5,1){$\bullet$}
\rput(2.5,0.7){$i_\ell$}\rput(3.5,0.7){$\;i_{1}$}
\rput(4,1){$\bullet$}\rput(5,1){$\bullet$}
\rput(4,0.7){$i_1$}\rput(5,0.7){$\;i_{\ell}$}
\end{pspicture}\quad\right)
\end{equation*}
}
\label{FigComTrace}
\end{figure}

\begin{figure}[h]\caption{ {\small Effect of an edge on trace-type ribbon graphs}}
\centering
{
\begin{equation*}
\sum_{i_{\ell}, i_{\ell+1}} W \left(\;
\psset{unit=0.8cm}\begin{pspicture}(1,1.5)(7,2.5)
 \psline[linewidth=.5pt ]{-}(1,1)(1,1.8)(2,1.8)(2,1)(1,1)\rput(1.5,1.4){$T_1$}
 \psline[linewidth=1pt,linearc=.3,linestyle=dotted]{->}(2,1) (2.4,1)
  \psline[linewidth=1pt,linearc=.3,linestyle=dotted]{->}(3,1) (3.4,1)
  \psline[linewidth=.5pt ]{-}(3.5,1)(3.5,1.8)(4.5,1.8)(4.5,1)(3.5,1)\rput(4,1.4){$T_2$}
   \psline[linewidth=1pt,linearc=.3,linestyle=dotted]{->}(4.5,1) (4.9,1)
    \psline[linewidth=1pt,linearc=.3,linestyle=dotted]{->}(5.5,1) (5.9,1)
\psline[linewidth=.5pt ]{-}(6,1)(6,1.8)(7,1.8)(7,1)(6,1)\rput(6.5,1.4){$T_3$}
\psline[linewidth=1pt,linearc=.4,ArrowInside=-<](5.5,1)(5.5,2.5)(2.5,2.5)(2.5,1)
\psline[linewidth=1pt,linearc=.3,ArrowInside=-<](3,1)(3,2)(5,2)(5,1)
\rput(1,1){$\bullet$}\rput(2,1){$\bullet$}\rput(2.5,1){$\bullet$}\rput(3,1){$\bullet$}\rput(3.5,1){$\bullet$}
\rput(4.5,1){$\bullet$}\rput(5,1){$\bullet$}\rput(5.5,1){$\bullet$}\rput(6,1){$\bullet$}\rput(7,1){$\bullet$}
\rput(1,0.7){$i_1$}\rput(2,0.7){$i_k$}
\rput(3,0.7){$\;\;i_{k+1}$}
\rput(4.5,0.7){$i_{\ell}$}
\rput(5.5,0.7){$\;\;i_{\ell+1}$}
\rput(7,0.7){$\,i_{1}$}
\end{pspicture}
\;\right)
= \hbar \sigma(i_{k+1})  W\left(\;
\psset{unit=0.8cm}\begin{pspicture}(11,1.5)(15,2.5)
\psline[linewidth=.5pt]{-}(11,1)(11,1.8)(12,1.8)(12,1)(11,1)\rput(11.5,1.4){$T_1$}
\psline[linewidth=1pt,linearc=.3,linestyle=dotted]{->}(12,1) (12.5,1)
\psline[linewidth=.5pt ]{-}(12.5,1)(12.5,1.8)(13.5,1.8)(13.5,1)(12.5,1)\rput(13,1.4){$T_3$}
\psline[linewidth=.5pt]{-}(14,1)(14,1.8)(15,1.8)(15,1)(14,1)\rput(14.5,1.4){$T_2$}
\rput(11,1){$\bullet$}\rput(12,1){$\bullet$}\rput(12.5,1){$\bullet$}\rput(13.5,1){$\bullet$}\rput(14,1){$\bullet$}\rput(15,1){$\bullet$}
\rput(11,0.7){$i_1$}\rput(12,0.7){$i_k$}
\rput(13.5,0.7){$i_1$}
\rput(14.2,0.7){$i_{k+1}$}
\rput(15,0.7){$i_{k+1}$}
\end{pspicture}\;\right)
\end{equation*}
}
\label{FigEffectEdgeTrace}
\end{figure}

\begin{figure}[h]\caption{  {\small  Effect of crossing edges on trace-type ribbon graphs}}
\centering
{\small
\begin{multline*}
\sum_{\substack{i_{p},i_{p+1}\\i_q,i_{q+1}}}
W \left(\;
\psset{unit=0.8cm}
\begin{pspicture}(1,1.5)(12,3.5)
 \psline[linewidth=.5pt ]{-}(1,1)(1,1.8)(2,1.8)(2,1)(1,1)\rput(1.5,1.4){$T_1$}
 \psline[linewidth=1pt,linearc=.3,linestyle=dotted]{->}(2,1) (2.4,1)
  \psline[linewidth=1pt,linearc=.3,linestyle=dotted]{->}(3,1) (3.4,1)
  \psline[linewidth=.5pt ]{-}(3.5,1)(3.5,1.8)(4.5,1.8)(4.5,1)(3.5,1)\rput(4,1.4){$T_2$}
   \psline[linewidth=1pt,linearc=.3,linestyle=dotted]{->}(4.5,1) (4.9,1)
   \psline[linewidth=1pt,linearc=.4,ArrowInside=->](5,1)(5,3.5)(10.5,3.5)(10.5,1)
\psline[linewidth=1pt,linearc=.3,ArrowInside=-<](5.5,1)(5.5,3)(10,3)(10,1)
\pscustom[linewidth=5pt,linecolor=white,fillstyle=solid,fillcolor=white]{
\psline[linewidth=1pt,linearc=.4,ArrowInside=-<](8,1)(8,2.5)(2.5,2.5)(2.5,1)
  \psline[linewidth=1pt,linearc=.3,ArrowInside=-<](3,1)(3,2)(7.5,2)(7.5,1)
  }
   \psline[linewidth=1pt,linearc=.4,ArrowInside=-<](8,1)(8,2.5)(2.5,2.5)(2.5,1)
  \psline[linewidth=1pt,linearc=.3,ArrowInside=-<](3,1)(3,2)(7.5,2)(7.5,1)
    \psline[linewidth=1pt,linearc=.3,linestyle=dotted]{->}(5.5,1) (5.9,1)
\psline[linewidth=.5pt ]{-}(6,1)(6,1.8)(7,1.8)(7,1)(6,1)\rput(6.5,1.4){$T_3$}
    \psline[linewidth=1pt,linearc=.3,linestyle=dotted]{->}(7,1) (7.4,1)
        \psline[linewidth=1pt,linearc=.3,linestyle=dotted]{->}(8,1) (8.4,1)
        \psline[linewidth=.5pt ]{-}(8.5,1)(8.5,1.8)(9.5,1.8)(9.5,1)(8.5,1)\rput(9,1.4){$T_4$}
\psline[linewidth=1pt,linearc=.3,linestyle=dotted]{->}(9.5,1) (9.9,1)
\psline[linewidth=1pt,linearc=.3,linestyle=dotted]{->}(10.5,1) (10.9,1)
\psline[linewidth=.5pt ]{-}(11,1)(11,1.8)(12,1.8)(12,1)(11,1)\rput(11.5,1.4){$T_5$}
\rput(1,1){$\bullet$}\rput(2,1){$\bullet$}\rput(2.5,1){$\bullet$}\rput(3,1){$\bullet$}\rput(3.5,1){$\bullet$}
\rput(4.5,1){$\bullet$}\rput(5,1){$\bullet$}\rput(5.5,1){$\bullet$}\rput(6,1){$\bullet$}\rput(7,1){$\bullet$}
\rput(7.5,1){$\bullet$}\rput(8,1){$\bullet$}\rput(8.5,1){$\bullet$}\rput(9.5,1){$\bullet$}\rput(10,1){$\bullet$}
\rput(10.5,1){$\bullet$}\rput(11,1){$\bullet$}\rput(12,1){$\bullet$}
\rput(1,0.7){$i_1$}\rput(2,0.7){$i_k$}
\rput(3,0.7){$\;\;i_{k+1}$}
\rput(4.5,0.7){$i_{\ell}$}
\rput(5.5,0.7){$\;\;i_{\ell+1}$}
\rput(7,0.7){$\,i_{p}$}
\rput(8,0.7){$\,i_{p+1}$}
\rput(9.5,0.7){$\,i_{q}$}
\rput(10.5,0.7){$\,i_{q+1}$}
\rput(12,0.7){$\,i_{1}$}
\end{pspicture}
\;\right)\\
  = \hbar^2 W\left(\;
\psset{unit=0.8cm}
\begin{pspicture}(11,1.5)(18,2.5)
\psline[linewidth=.5pt]{-}(11,1)(11,1.8)(12,1.8)(12,1)(11,1)\rput(11.5,1.4){$T_1$}
\psline[linewidth=1pt,linearc=.3,linestyle=dotted]{->}(12,1) (12.4,1)
\psline[linewidth=.5pt ]{-}(12.5,1)(12.5,1.8)(13.5,1.8)(13.5,1)(12.5,1)\rput(13,1.4){$T_4$}
\psline[linewidth=1pt,linearc=.3,linestyle=dotted]{->}(13.5,1) (13.9,1)
\psline[linewidth=.5pt]{-}(14,1)(14,1.8)(15,1.8)(15,1)(14,1)\rput(14.5,1.4){$T_3$}
\psline[linewidth=1pt,linearc=.3,linestyle=dotted]{->}(15,1) (15.4,1)
\psline[linewidth=.5pt]{-}(15.5,1)(15.5,1.8)(16.5,1.8)(16.5,1)(15.5,1)\rput(16,1.4){$T_2$}
\psline[linewidth=1pt,linearc=.3,linestyle=dotted]{->}(16.5,1) (16.9,1)
\psline[linewidth=.5pt ]{-}(17,1)(17,1.8)(18,1.8)(18,1)(17,1)\rput(17.5,1.4){$T_5$}
\rput(11,1){$\bullet$}\rput(12,1){$\bullet$}\rput(12.5,1){$\bullet$}\rput(13.5,1){$\bullet$}\rput(14,1){$\bullet$}\rput(15,1){$\bullet$}
\rput(15.5,1){$\bullet$}\rput(16.5,1){$\bullet$}\rput(17,1){$\bullet$}\rput(18,1){$\bullet$}
\rput(11,0.7){$i_1$}\rput(12,0.7){$i_k$}
\rput(13.5,0.7){$i_{\ell+1}$}
\rput(15,0.7){$i_{k+1}$}
\rput(16.5,0.7){$i_{\ell}$}
\rput(18,0.7){$i_{1}$}
\end{pspicture}\;\right)
\end{multline*}
}
\label{FigCrossingEdgesTrace}
\end{figure}

Let us consider a few examples in relation to the expectation value of $\str M^n$.  First of all, one easily verifies that the contribution of the trace-type graph of Eq.\eqref{eqGraphT}  is equal to $\hbar Y_{i_1i_2}\hbar Y_{i_2i_3}\cdots \hbar Y_{i_n i_1}$,
so summing over all the indices as in Eq.\eqref{EqFormulaExpectation} leads to $\str (\hbar Y)^n$.  Consider next  the contribution of a graph with one edge:
{\small \begin{equation*}\label{EqGraphTkl}
\psset{unit=0.8cm}
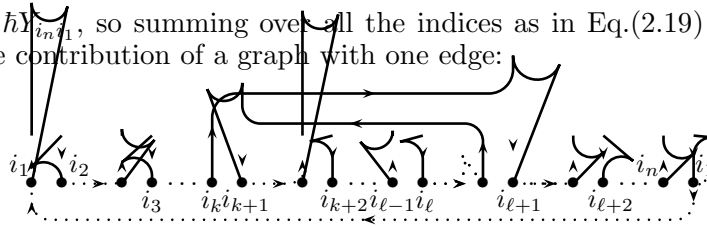
\begin{pspicture}(0,0)(12,2.5)
 \psline[linewidth=1pt,linearc=.3,ArrowInside=->](0.5,1)(0.5,1.8)
\psline[linewidth=1pt,linearc=.3,ArrowInside=->](1,1.8)(1,1)
\psline[linewidth=1pt,linearc=.3,ArrowInside=->](2,1)(2,1.8)
\psline[linewidth=1pt,linearc=.3,ArrowInside=->](2.5,1.8)(2.5,1)
\psline[linewidth=1pt,linearc=.3,ArrowInside=->](5,1)(5,1.8)
\psline[linewidth=1pt,linearc=.3,ArrowInside=->](5.5,1.8)(5.5,1)
\psline[linewidth=1pt,linearc=.3,ArrowInside=->](6.5,1)(6.5,1.8)
\psline[linewidth=1pt,linearc=.3,ArrowInside=->](7,1.8)(7,1)
\psline[linewidth=1pt,linearc=.3,ArrowInside=->](9.5,1)(9.5,1.8)
\psline[linewidth=1pt,linearc=.3,ArrowInside=->](10,1.8)(10,1)
\psline[linewidth=1pt,linearc=.3,ArrowInside=->](11,1)(11,1.8)
\psline[linewidth=1pt,linearc=.3,ArrowInside=->](11.5,1.8)(11.5,1)
 \psline[linewidth=1pt,linearc=.3,ArrowInside=->,linestyle=dotted](11.5,1)(11.5,0.4)(0.5,0.4)(0.5,1)
 \psline[ArrowInside=-<,linewidth=1pt,linearc=.3,linestyle=dotted](2,1)(1,1)
 \psline[ArrowInside=-<,linewidth=1pt,linearc=.3,linestyle=dotted](5,1)(4,1)
 \psline[ArrowInside=-<,linewidth=1pt,linearc=.3,linestyle=dotted](8,1)(7,1)
  \psline[ArrowInside=-<,linewidth=1pt,linearc=.3,linestyle=dotted](8,1)(7,1)
   \psline[ArrowInside=-<,linewidth=1pt,linearc=.3,linestyle=dotted](9.5,1)(8.5,1)
\psline[linewidth=1pt,linearc=.4,ArrowInside=-<](8.5,1)(8.5,2.5)(3.5,2.5)(3.5,1)
\psline[linewidth=1pt,linearc=.3,ArrowInside=-<](4,1)(4,2)(8,2)(8,1)
\rput(0.5,1){$\bullet$}\rput(1,1){$\bullet$}\rput(2,1){$\bullet$}\rput(2.5,1){$\bullet$}\rput(3.5,1){$\bullet$}\rput(4,1){$\bullet$}
\rput(8.5,1){$\bullet$}\rput(5,1){$\bullet$}\rput(8.0,1){$\bullet$}\rput(5.5,1){$\bullet$}\rput(6.5,1){$\bullet$}\rput(7,1){$\bullet$}\rput(9.5,1)
{$\bullet$}\rput(10,1){$\bullet$}\rput(11,1){$\bullet$}\rput(11.5,1){$\bullet$}
\rput(0.3,1.3){$i_1$}\rput(1.3,1.3){$i_2$}
\rput(2.5,0.7){$i_3$}\rput(3,1){$\ldots$}\rput(3.5,0.7){$i_k$}
\rput(4,0.7){$\;i_{k+1}$}
\rput(5.7,0.7){$i_{k+2}$}\rput(6,1){$\ldots$}
\rput(6.5,0.7){$i_{\ell-1}$}
\rput(7,0.7){$\,i_{\ell}$}
\rput(8.5,0.7){$\,\,\,i_{\ell+1}$}
\rput(9,1){$\ldots$}
\rput(10,0.7){$\;\;i_{\ell+2}$}\rput(10.5,1){$\ldots$}
\rput(10.8,1.3){$i_n\;$}\rput(11.7,1.3){$i_1$}
\end{pspicture}
\end{equation*}}
\noindent From Fig.\ref{FigEffectEdgeTrace}, we conclude if $T_2$ is equal to a sequence of non-connected half-edges of trace-type, then
{\small \begin{equation*}
\sum_{i_{k+1},\ldots, i_{\ell+1}} W \left(\;\;
\psset{unit=0.8cm}\begin{pspicture}(1,1.5)(7,2.5)
 \psline[linewidth=.5pt]{-}(1,1)(1,1.8)(2,1.8)(2,1)(1,1)\rput(1.5,1.4){$T_1$}
  \psline[linewidth=.5pt]{-}(3.5,1)(3.5,1.8)(4.5,1.8)(4.5,1)(3.5,1)\rput(4,1.4){$T_2$}
\psline[linewidth=.5pt]{-}(6,1)(6,1.8)(7,1.8)(7,1)(6,1)\rput(6.5,1.4){$T_3$}
\psline[linewidth=1pt,linearc=.4,ArrowInside=-<](5.5,1)(5.5,2.5)(2.5,2.5)(2.5,1)
\psline[linewidth=1pt,linearc=.3,ArrowInside=-<](3,1)(3,2)(5,2)(5,1)
\rput(1,1){$\bullet$}\rput(2,1){$\bullet$}\rput(2.5,1){$\bullet$}\rput(3,1){$\bullet$}\rput(3.5,1){$\bullet$}
\rput(4.5,1){$\bullet$}\rput(5,1){$\bullet$}\rput(5.5,1){$\bullet$}\rput(6,1){$\bullet$}\rput(7,1){$\bullet$}
\rput(1,0.7){$i_1$}\rput(2,0.7){$i_k$}
\rput(3,0.7){$\;i_{k+1}$}
\rput(4.5,0.7){$i_{\ell}$}
\rput(5.5,0.7){$\;i_{\ell+1}$}
\rput(7,0.7){$\,i_{1}$}
\psline[linewidth=1pt,linearc=.3,linestyle=dotted]{->}(2,1) (2.4,1)
\psline[linewidth=1pt,linearc=.3,linestyle=dotted]{->}(3,1) (3.4,1)
\psline[linewidth=1pt,linearc=.3,linestyle=dotted]{->}(4.5,1) (4.9,1)
\psline[linewidth=1pt,linearc=.3,linestyle=dotted]{->}(5.5,1) (5.9,1)
\end{pspicture}
\;\;\right)\, =\, \hbar \str (\hbar Y)^{\ell-k-1}\, W\left(\;\;
\psset{unit=0.8cm}\begin{pspicture}(11,1.5)(13.5,2.5)
\psline[linewidth=.5pt]{-}(11,1)(11,1.8)(12,1.8)(12,1)(11,1)\rput(11.5,1.4){$T_1$}
\psline[linewidth=.5pt]{-}(12.5,1)(12.5,1.8)(13.5,1.8)(13.5,1)(12.5,1)\rput(13,1.4){$T_3$}
\rput(11,1){$\bullet$}\rput(12,1){$\bullet$}\rput(12.5,1){$\bullet$}\rput(13.5,1){$\bullet$}
\rput(11,0.7){$i_1$}\rput(12,0.7){$i_k$}
\rput(13.5,0.7){$i_1$}
\psline[linewidth=1pt,linearc=.3,linestyle=dotted]{->}(12,1) (12.4,1)
\end{pspicture}\;\;\right).
\end{equation*}
}

\noindent Note that in the case where $\ell=k+1$, which corresponds to graph with a closed loop labeled by $i_{k+1}$,
 the previous formula remains valid if we interpret $\str (\hbar Y)^{\ell-k-1}$ as $\str \mathbf{1}=p-q$.
   We then conclude that the total contribution to the expectation value of $\str M^n$ of all  labeled graphs $T_{k\ell}$ as in Eq.\eqref{EqGraphTkl},
 with one edge linking the half-edges $i_ki_{k+1}$ and $i_{\ell}i_{\ell+1}$, is
 $$\sum_{i_1,\ldots,i_n}\sigma(i_1) W(T_{k\ell})=\hbar\str (\hbar Y)^{n+k-\ell-1}\str (\hbar Y)^{\ell-k-1}.$$

The ribbon graph $T_{15,27,36}$  given in Fig.\ref{figTr} is more complicated since the calculation of its contribution to the expectation of $\str M^8$ requires the use of the three rules given in Figures \ref{FigComTrace}, \ref{FigEffectEdgeTrace}, and \ref{FigCrossingEdgesTrace}.
One gets  $$\sum_{1\leq i_{1},\ldots, i_{8}\leq p+q} \sigma(i_1)W \left(T_{15,27,36}\right)=\hbar^3(p-q)\str(\hbar Y)^2.$$
The different factors in the last equation can be understood as follows: the 3 edges create a  $\hbar^3$; $p-q$ comes from the loop starting at $i_3$; $\str (\hbar Y)^2$ is obtained when considering the two non-connected half-edges.


\begin{figure}[h]\caption{  {\small Possible ribbon graph associated  to $\langle \mathrm{str}
M^8\rangle$}}
\centering
{ \psset{unit=0.8cm}
\begin{pspicture}(0,0)(12,5.5)
\psline[linewidth=1pt,linearc=.3,ArrowInside=->](11,1)(11,1.6)
\psline[linewidth=1pt,linearc=.3,ArrowInside=->](11.5,1.6)(11.5,1)
\psline[linewidth=1pt,linearc=.3,ArrowInside=->](5,1)(5,1.6)
\psline[linewidth=1pt,linearc=.3,ArrowInside=->](5.5,1.6)(5.5,1)
 \psline[linewidth=1pt,linearc=.3,ArrowInside=->,linestyle=dotted](11.5,1)(11.5,0.4)(0.5,0.4)(0.5,1)
 \psline[ArrowInside=-<,linewidth=1pt,linearc=.3,linestyle=dotted](2,1)(1,1)
 \psline[ArrowInside=-<,linewidth=1pt,linearc=.3,linestyle=dotted](3.5,1)(2.5,1)
 \psline[ArrowInside=-<,linewidth=1pt,linearc=.3,linestyle=dotted](5,1)(4,1)
 \psline[ArrowInside=-<,linewidth=1pt,linearc=.3,linestyle=dotted](6.5,1)(5.5,1)
 \psline[ArrowInside=-<,linewidth=1pt,linearc=.3,linestyle=dotted](8,1)(7,1)
  \psline[ArrowInside=-<,linewidth=1pt,linearc=.3,linestyle=dotted](8,1)(7,1)
   \psline[ArrowInside=-<,linewidth=1pt,linearc=.3,linestyle=dotted](9.5,1)(8.5,1)
  \psline[ArrowInside=-<,linewidth=1pt,linearc=.3,linestyle=dotted](11,1)(10,1)
\psline[ArrowInside=-<,linewidth=1pt,linearc=.3](2.5,1)(2.5,3)(9.5,3)(9.5,1)
\psline[linewidth=1pt,linearc=.4,ArrowInside=-<](10,1)(10,3.5)(2,3.5)(2,1)
\pscustom[linewidth=5pt,linecolor=white,fillstyle=solid,fillcolor=white]{%
\psline[linewidth=1pt,linearc=.4,ArrowInside=-<](7,1)(7,2.5)(0.5,2.5)(0.5,1)
\psline[linewidth=1pt,linearc=.3,ArrowInside=-<](1,1)(1,2)(6.5,2)(6.5,1)
}
\psline[linewidth=1pt,linearc=.4,ArrowInside=-<](7,1)(7,2.5)(0.5,2.5)(0.5,1)
\psline[linewidth=1pt,linearc=.3,ArrowInside=-<](1,1)(1,2)(6.5,2)(6.5,1)
\pscustom[linewidth=5pt,linecolor=white,fillstyle=solid,fillcolor=white]{%
\psline[linewidth=1pt,linearc=.4,ArrowInside=-<](8.5,1)(8.5,4.5)(3.5,4.5)(3.5,1)
\psline[linewidth=1pt,linearc=.3,ArrowInside=-<](4,1)(4,4)(8,4)(8,1)
}
\psline[linewidth=1pt,linearc=.4,ArrowInside=-<](8.5,1)(8.5,4.5)(3.5,4.5)(3.5,1)
\psline[linewidth=1pt,linearc=.3,ArrowInside=-<](4,1)(4,4)(8,4)(8,1)
\pscustom[linewidth=6pt,linecolor=white,fillstyle=solid,fillcolor=white]{%
 \psline[linewidth=1pt,linearc=.3]{->}(5,1)(5,5)
\psline[linewidth=1pt,linearc=.4]{-<}(5.5,5)(5.5,1) }
 \psline[linewidth=1pt,linearc=.3]{->}(5,1)(5,5)
\psline[linewidth=1pt,linearc=.4]{>-}(5.5,5)(5.5,1)
\pscustom[linewidth=6pt,linecolor=white,fillstyle=solid,fillcolor=white]{%
 \psline[linewidth=1pt,linearc=.3]{->}(11,1)(11,5)
\psline[linewidth=1pt,linearc=.4]{>-}(11.5,5)(11.5,1) }
 \psline[linewidth=1pt,linearc=.3]{->}(11,1)(11,5)
\psline[linewidth=1pt,linearc=.4]{>-}(11.5,5)(11.5,1)
\rput(0.5,1){$\bullet$}\rput(1,1){$\bullet$}\rput(2,1){$\bullet$}\rput(2.5,1){$\bullet$}\rput(3.5,1){$\bullet$}\rput(4,1){$\bullet$}
\rput(8.5,1){$\bullet$}\rput(5,1){$\bullet$}\rput(8.0,1){$\bullet$}\rput(5.5,1){$\bullet$}\rput(6.5,1){$\bullet$}\rput(7,1){$\bullet$}\rput(9.5,1)
{$\bullet$}\rput(10,1){$\bullet$}\rput(11,1){$\bullet$}\rput(11.5,1){$\bullet$}
\rput(0.3,1.3){$i_1$}\rput(1.3,1.3){$i_2$}
\rput(2.5,0.7){$i_3$}
\rput(4,0.7){$i_4$}
\rput(5.5,0.7){$i_5$}
\rput(7,0.7){$i_6$}
\rput(8.5,0.7){$i_7$}
\rput(10,0.7){$i_8$}
\rput(11.7,1.3){$i_1$}
\end{pspicture}
 \label{figTr}}
\end{figure}
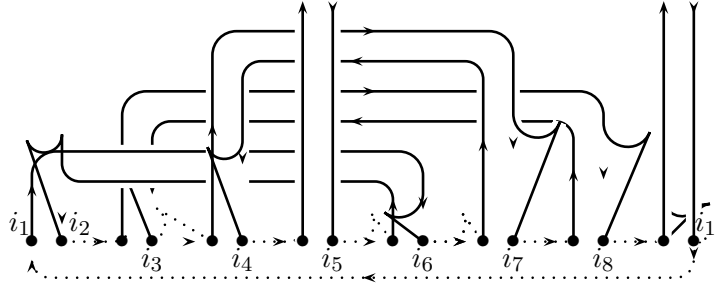

The previous discussion suggests that the weight associated to the sum over the indices of a trace-type ribbon graph depends only $\hbar$ and
 supertraces of the form $\str (\hbar Y)^m$, which includes the contribution $\str (\hbar Y)^0=\str \mathbf{1}=p-q$.
 It is indeed the case since  Lema \ref{LemmaTraceGraph} implies that when reordering the subgraphs correctly,
 the only signs that appear in the total weight are factors of the form  $\sigma(i_k)$, each of them associated to an independent trace-type subgraph.
 This means that each independent subgraph produces a factor $p-q$ if it is empty while it produces a factor $\str (\hbar Y)^n$ if it is made of
 $n$ non-connected half-edges.  Moreover, each edge produce a factor $\hbar$.

We can go further by defining a face of a trace-type ribbon graph as the region contained inside a loop,
which is a closed path that follows  the arrows (including the dotted arrow).
Note that a face can contain an independent  subgraph that contains non-connected half-edges.
 For instance, the graph given in Fig.\ref{figTr} has two faces: the first is delimited by the path $i_1\to i_6\to i_4\to i_5\to i_2\to i_8\to i_1$;
 the second  is found by following the path $i_3\to i_7\to i_3$.  By making use of Lemma \ref{LemmaTraceGraph},
 we conclude that each face, when summing over all the indices, carries a weight of $\str (\hbar Y)^m$, where $m\geq 0$
 corresponds to the number of non-connected half-edges.  The weight of the whole graph is given by the product of the weights associated
  to the edges and the faces.   Thus, if the labeled trace-type ribbon-graph $T$ that contains $n-m$ edges and $F$ faces $f_i$,
  each of them including $m_i\geq 0$ non-connected half-edges,
\beq
\sum_{i_1,\ldots,i_{n}}\sigma(i_1)W(T)=\hbar^n\prod_{i=1}^F\str(\hbar Y)^{m_i}
\eeq
Now define $m=\sum_im_n$ as the total number of non-connected half-edges in $T$.

Obviously, to any trace-type ribbon graph $T$  with $n-m$ edges, $m$ half-edges,
$F$ faces and labeled by the integers $i_1,\ldots,i_n$, we can associate a connected ribbon graph $\Gamma$, called a fatgraph star,
with $F$ faces, $V=1+m$ vertices, and $E=n$  edges labeled from 1 to n. Note that in the star $\Gamma$,
all the points carrying the indices of $T$ are considered as forming a unique vertex while the
$m$ non-connected half-edges of $T$ are interpreted as being connected to $m$ distinct vertices,
so that the total number of vertices in $\Gamma$ is $1+m$.  For example, the fatgraph star associated to the trace-type ribbon graph of Fig.\ref{figTr}
 is given in Fig.\ \ref{StarFatGraph}.  It is a classical result that any connected graph with $F$ faces, $E$ edges, and $V$ vertices, can be embedded on a surface of genus (at least) $g$ in a way that guaranties that the edges don't cross (except at the vertices) if $g$ is given by the following expression of the Euler characteristic \cite{Zvonkin}:
\beq
\chi=V+F-E=2-2g
\eeq
Thus, if $T$ is a trace-type ribbon graph with $m$ non-connected half-edges, $n-m$ edges, and $F$ faces,
\beq
\hbar^{-1}\sum_{i_1,\ldots,i_{n}}\sigma(i_1)W(T)=\hbar^{n-m-1-F}\prod_{i=1}^F\hbar\str(\hbar Y)^{m_i}
\eeq
which yields, when $E=n$ and $V=1+m$,
\beq
\hbar^{-1}\sum_{i_1,\ldots,i_{n}}\sigma(i_1)W(T)=\hbar^{E-V-F}\prod_{i=1}^F\hbar\str(\hbar Y)^{m_i}=\hbar^{-2+2g}\prod_{i=1}^F\hbar\str(\hbar Y)^{m_i}
\eeq
Going back to Eq.\eqref{EqFormulaExpectation}, we conclude that if
\beq \bar Y=\hbar Y,\quad p_{m_i}=\hbar \str \bar Y^{m_i},\quad\eeq
then
\beq\label{EqResultExpectation}\ee{-\frac{1}{2\hbar^2}p_2} \Big\langle \frac{1}{\hbar}
\str M^n \ee{\frac{1}{\hbar}\str M\bar Y}\Big\rangle_{M\in GU(p|q)}= \sum_\Gamma \hbar^{-2+2g}\prod_i p_{m_i}.
\eeq
where the sum runs over all connected fatgraph stars  $\Gamma$ with one $n$-valent vertex, $m$ univalent vertices and $F$ faces,
each face containing $m_i$ edges connected to $m_i$ distinct vertices.


\begin{figure}[h]\caption{  {\small The star fatgraph associated to the ribbon graph of Fig.\ \ref{figTr}.
The half-edges are labeled from 1 to 8.   The graph contains $V=3$ vertices (denoted  $a,b,c$), $F=2$ faces (denoted $F_1, F_2$), and $E=5$ edges.  As shown on the right-hand side, it can be embedded on compact surface of genus $g=(E-F-V+2)/2=1$.}   }
\centering
{ \psset{unit=0.8cm}
\begin{pspicture}(-3,-4)(4.5,5)
 \rput(0,1){$\bullet$}\rput(0.25,1.5){$1$}\rput(0.5,1){$\bullet$}\rput(1,1.2){$2$}\rput(0.8,0.6){$\bullet$}\rput(1.5,.3){$3$}\rput(0.8,0.1){$\bullet$}
 \rput(1,-.4){$4$}\rput(0.5,-0.3){$\bullet$}\rput(.25,-1){$5$}\rput(0,-0.3){$\bullet$}\rput(-.5,-.4){$6$}
 \rput(-0.3,0.1){$\bullet$} \rput(-1,.3){$7$}\rput(-0.3,0.6){$\bullet$}\rput(-.6,1.2){$8$}
 \rput(-1,4){$F_1$}
 \rput(-1,2.5){$F_2$}
 \rput(1,-1.5){$F_2$}
 \rput(0.25,0.5){a}
\psline[linewidth=1pt,linearc=.5,ArrowInside=->](0,1)(0,2)
\psline[linewidth=1pt,linearc=.4,ArrowInside=->](0.5,1)(1.5,2)(1.5,3)(-2,3)(-2,0.6)(-1.3,0.6)
\psline[linewidth=1pt,linearc=.5,ArrowInside=->](-1.3,0.1)(-2.5,0.1)(-2.5,3.5)(2,3.5)(2,1.8)(0.8,0.6)
\psline[linewidth=1pt,linearc=.3,ArrowInside=->](0.8,0.6)(1.8,0.6)
\psline[linewidth=1pt,linearc=.3,ArrowInside=->](1.8,0.1)(0.8,0.1)
\psline[linewidth=1pt,linearc=.3,ArrowInside=->](0.8,0.1)(1.8,-0.9)
\psline[linewidth=1pt,linearc=.3,ArrowInside=->](1.5,-1.3)(0.5,-0.3)
\rput(1.7,-1.1){c}
\psline[linewidth=1pt,linearc=.3,ArrowInside=->](0.5,-0.3)(0.5,-1.3)
\psline[linewidth=1pt,linearc=.3,ArrowInside=->](0,-1.3)(0,-0.3)
\psline[linewidth=1pt,linearc=.4,ArrowInside=->](0,-0.3)(-1,-1.3)(-1,-2)(2.5,-2)(2.5,0.1)(1.8,.1)
\psline[linewidth=1pt,linearc=.5,ArrowInside=->](1.8,0.6)(3,0.6)(3,-2.5)(-1.5,-2.5)(-1.5,-1.1)(-0.3,0.1)
\psline[linewidth=1pt,linearc=.3,ArrowInside=->](-0.3,0.1)(-1.3,0.1)
\psline[linewidth=1pt,linearc=.3,ArrowInside=->](-1.3,0.6)(-0.3,0.6)

\psline[linewidth=1pt,linearc=.3,ArrowInside=->](-0.3,0.6)(-1.3,1.6)
\psline[linewidth=1pt,linearc=.3,ArrowInside=->](-1,2)(0,1)

\rput(-1.2,1.9){b}

\pscustom[linewidth=5pt,linecolor=white,fillstyle=solid,fillcolor=white]{
\psline[linewidth=1pt,linearc=.5,ArrowInside=->](0,2)(0,4.5)(4,4.5)(4,-3.5)(0,-3.5)(0,-1.3)
\psline[linewidth=1pt,linearc=.4,ArrowInside=->](.5,-1.3)(.5,-3)(3.5,-3)(3.5,4)(.5,4)(0.5,2)
}
\psline[linewidth=1pt,linearc=.5,ArrowInside=->](0,1)(0,2)(0,4.5)(4,4.5)(4,-3.5)(0,-3.5)(0,-1.3)(0,-.3)
\psline[linewidth=1pt,linearc=.4,ArrowInside=->](.5,-1.3)(.5,-3)(3.5,-3)(3.5,4)(.5,4)(0.5,2)(0.5,1)

\end{pspicture}
}
{ \psset{unit=0.7cm}
\begin{pspicture}(0,-3)(12,6)
\pscustom[linewidth=2pt,fillstyle=solid,fillcolor=lightgray]{
\psellipse(6,3)(6,2.4)}
\pscustom[linewidth=0pt,fillstyle=solid,fillcolor=white]{
\pscurve(4,3.3)(6,3.5)(8,3.3)
\pscurve(8,3.3)(6,3)(4,3.3)
}
\pscurve[linewidth=2pt](3,3.5)(6,3)(9,3.5)
\pscurve[linewidth=2pt](4,3.3)(6,3.5)(8,3.3)
\psellipse[doubleline=true,doublesep=2pt](6,3.3)(4.6,1.3)
\pscurve[linewidth=1pt,doubleline=true,doublesep=2pt,linearc=.4](4,2.2)(3,3)(3.5,3.3)
\pscurve[linestyle=dashed,dash=3pt](3.5,3.4)(4.4,2)(4,.8)
\pscurve[linewidth=1pt,doubleline=true,doublesep=2pt,linearc=.4](1,2.6)(2,2)(4,2.2)
\pscurve[linewidth=1pt,doubleline=true,doublesep=2pt,linearc=.2](7,0.7)(4,2.2)(6,2.9)
\pscurve[linestyle=dashed,dash=3pt](6,2.9)(9,1.6)(7,0.7)
\pscurve[linewidth=1pt,doubleline=true,doublesep=2pt,linearc=.2](4,2.2)(7,2.5)(9,3)
\pscurve[linewidth=1pt,doubleline=true,doublesep=2pt,linearc=.4](4,2.2)(3.6,1.6)(4,.8)

\pscurve[linewidth=2pt](8,3.3)(6,3)(4,3.3)

\pscustom[linewidth=2pt]{
\pscurve[linewidth=1pt](3,3.5)(6,3)(9,3.5)
}
\pscustom[linewidth=2pt]{
\psellipse(6,3)(6,2.4)}
\pscircle*(4,2.2){.15}
\pscircle*(1,2.6){.15}
\pscircle*(9,3){.15}
\rput(1,3.6){$F_2$}
\rput(4.6,1){$F_1$}
\end{pspicture}
}
 \label{StarFatGraph}
\end{figure}
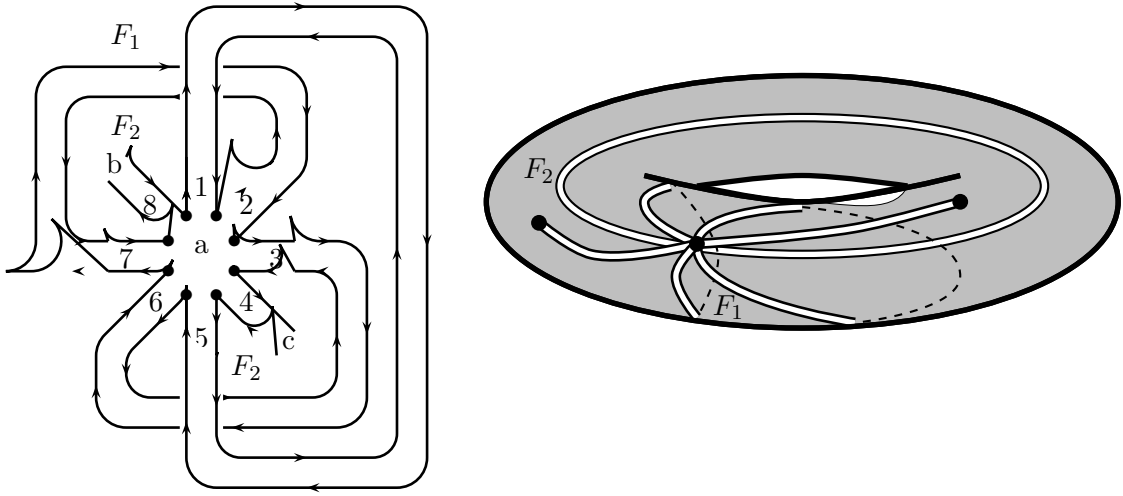

Any correlation of product of supertraces $\str M^n$ can be evaluated with
the help of a simple generalization the combinatorial method exposed above.
Indeed, Proposition \ref{PropBijection} implies that

\begin{multline}
\ee{-\frac{1}{2\hbar}\str \bar Y^2} \Big\langle \prod_{k=1}^d\frac{1}{\hbar}
\str M^{n_k} \ee{\frac{1}{\hbar}\str M\bar Y}\Big\rangle_{M\in GU(p|q)}\\
=\sum_{1\leq k\leq d\phantom{^k}}\,\sum_{1\leq i^k_1,\ldots, i^k_n\leq p+q} \sigma(i_1) \cdots \sigma(i_d)\, W(\,\Sigma T^{n_1}\cdots T^{n_d}\,)
\end{multline}
where $\Sigma T^{n_1}\cdots T^{n_d}$ stands for the sum of all possible ribbon graphs that can be obtained by connecting (partially or totally)
the half-edges contained in the following concatenation of $d$ trace-type ribbon graphs (cf.  Eq.\ref{eqGraphT}):
{\small \beq\label{eqGraphProdT}
\psset{unit=0.85cm}\begin{pspicture}(0,0.5)(6,1.8)
\psline[linewidth=1pt,linearc=.3]{->}(0.5,1)(0.5,1.8)
\psline[linewidth=1pt,linearc=.4]{>-}(1,1.8)(1,1)
\psline[linewidth=1pt,linearc=.3]{->}(2,1)(2,1.8)
\psline[linewidth=1pt,linearc=.4]{>-}(2.5,1.8)(2.5,1)
 \psline[linewidth=1pt,linearc=.3]{->}(4,1)(4,1.8)
\psline[linewidth=1pt,linearc=.4]{>-}(4.5,1.8)(4.5,1)
 \psline[linewidth=1pt,linearc=.3]{->}(5.5,1)(5.5,1.8)
\psline[linewidth=1pt,linearc=.4]{>-}(6,1.8)(6,1)
 \psline[linewidth=1pt,linearc=.3,ArrowInside=->,linestyle=dotted](4.5,1)(5.5,1)
\psline[linewidth=1pt,linearc=.3,ArrowInside=->,linestyle=dotted](6,1)(6,0.4)(0.5,0.4)(0.5,1)
 \psline[linewidth=1pt,linearc=.3,ArrowInside=->,linestyle=dotted](1,1)(2,1)
  \rput(3.3,1){$\ldots$}
\rput(0.5,1){$\bullet$}\rput(1,1){$\bullet$}\rput(2,1){$\bullet$}\rput(2.5,1){$\bullet$}\rput(4,1){$\bullet$}\rput(4.5,1){$\bullet$}
\rput(6,1){$\bullet$}\rput(5.5,1){$\bullet$}
\rput(0.3,1.3){$i_1^1$}\rput(1.3,1.3){$i_2^1$}\rput(2,0.7){$i_2^1$}\rput(2.5,0.7){$i_3^1$}\rput(4,0.7){$i_{n_1-1}^1$}\rput(4.7,0.7){$i_{n_1}^1$}
\rput(5.3,1.3){$i_{n_1}^1$}\rput(6.2,1.3){$i_1^1$}
\end{pspicture}
\qquad\qquad
\psset{unit=0.85cm}
\begin{pspicture}(0,0.5)(6,1.8)
\rput(-0.5,1){$\ldots$}
\psline[linewidth=1pt,linearc=.3]{->}(0.5,1)(0.5,1.8)
\psline[linewidth=1pt,linearc=.4]{>-}(1,1.8)(1,1)
\psline[linewidth=1pt,linearc=.3]{->}(2,1)(2,1.8)
\psline[linewidth=1pt,linearc=.4]{>-}(2.5,1.8)(2.5,1)
 \psline[linewidth=1pt,linearc=.3]{->}(4,1)(4,1.8)
\psline[linewidth=1pt,linearc=.4]{>-}(4.5,1.8)(4.5,1)
 \psline[linewidth=1pt,linearc=.3]{->}(5.5,1)(5.5,1.8)
\psline[linewidth=1pt,linearc=.4]{>-}(6,1.8)(6,1)
 \psline[linewidth=1pt,linearc=.3,ArrowInside=->,linestyle=dotted](4.5,1)(5.5,1)
\psline[linewidth=1pt,linearc=.3,ArrowInside=->,linestyle=dotted](6,1)(6,0.4)(0.5,0.4)(0.5,1)
 \psline[linewidth=1pt,linearc=.3,ArrowInside=->,linestyle=dotted](1,1)(2,1)
  \rput(3.3,1){$\ldots$}
\rput(0.5,1){$\bullet$}\rput(1,1){$\bullet$}\rput(2,1){$\bullet$}\rput(2.5,1){$\bullet$}\rput(4,1){$\bullet$}\rput(4.5,1){$\bullet$}
\rput(6,1){$\bullet$}\rput(5.5,1){$\bullet$}
\rput(0.3,1.3){$i_1^d$}\rput(1.3,1.3){$i_2^d$}\rput(2,0.7){$i_2^d$}\rput(2.5,0.7){$i_3^d$}\rput(4,0.7){$i_{n_d-1}^d$}\rput(4.7,0.7){$i_{n_d}^d$}
\rput(5.3,1.3){$i_n^d$}\rput(6.2,1.3){$i_1^d$}
\end{pspicture}
\eeq}
In the last equation, the $i$th independent subgraph has valency $n_i$.
The connection process may produce a minimum of 0 and a  maximum of $d$ non-connected trace-type ribbon graphs.  The use of Lemma \ref{LemmaTraceGraph}
allows to eliminate the sum over the indices $i^k_\ell$ and leads to the following proposition which establish that any correlation function as an expansion in terms of $(p-q)^{-2}$
and supertraces of the external field $Y$.

\begin{theorem}\label{PropCentral}
Let $p_n=\hbar \str   Y^n$, where $ Y$ denotes a Hermitian supermatrix or size $p+q$.
Suppose that $\Gamma$ is a star fatgraph with $d$ labeled vertices, each of them having a valency $n_i$, and define $E=\sum_{i=1}^d n_i$.
Suppose moreover that $\Gamma$ contains $\delta$ non-connected subgraphs $\Gamma_i$.  Define
the number of faces in $\Gamma$ as $F=\sum_{i=1}^\delta f_i$, where $f_i$ is equal to the the number of faces in $\Gamma_i$.
Finally, let $m_i$ stand for the number of non-connected half-edges in the $i$th face of $\Gamma$ and  $V=d+\sum_{i=1}^Fm_i$.
Then
\beq
 \Big\langle \prod_{k=1}^d\frac{1}{\hbar}
\str M^{n_k} \ee{\frac{1}{\hbar}\str M  Y}\Big\rangle_{M\in GU(p|q)}=\ee{\frac{1}{2\hbar^2}p_2}\,\sum_\Gamma \hbar^{E-V-F}\prod_{i=1}^F p_{m_i}
\eeq
where the sum extends over all possible (connected or not) star fatgraphs $\Gamma$.  Moreover, $V+F-E=2-2g(\Gamma)$, where
$g(\Gamma)=\sum_{i=1}^\delta g_i$ and where $g_i$ denotes the minimal genus of the compact surface on which the graph $\Gamma_i$ can be embedded.
\end{theorem}


We remark that if $\hbar=(p-q)^{-1}$, then the polynomial $p_0=1$.  In that case,
the previous theorem directly implies that when the external field is zero, the correlation functions of
 a Gaussian supermatrix model are equivalent to that of a Gaussian model involving a $N\times N$ matrix if one identifies $\hbar$ with $N$.

\begin{corollary}\label{PropCentral}
Let $\hbar=(p-q)^{-1}$ and let $n=\sum_k n_k$ be an even integer.  Then, with the notation explained above,
\beq
 \Big\langle \prod_{k=1}^d\frac{1}{\hbar}
\str M^{n_k} \ee{\frac{1}{\hbar}\str M  Y}\Big\rangle_{M\in GU(p|q)}= \,\sum_\Gamma \hbar^{2g(\Gamma)-2}
\eeq
where the sum extends over all possible (connected or not) star fatgraphs $\Gamma$  in which the $n$  half-edges are all connected by pairs.
\end{corollary}

Let us make a last remark regarding the partition of the
one-supermatrix model \beq \label{EqPartFormal}
\mathcal{Z}_{(p|q)}(t_1,t_2,\ldots)=\left\langle \ee{-\sum_{k\geq 1}
\frac{t_k}{k}\str M^k}\right\rangle_{M\in GU{(p|q)}} \eeq This is
the usual form of the generating function considered for the
enumeration of maps.  As explained in \cite{Des}, partition
functions of one-matrix models, when considering the $t_k$ as formal
parameters (thus,  not as complex numbers), are equivalent to the
expectation value of products of characteristic polynomials in
Gaussian ensembles. This is easily generalized to the supermatrix
case.

\begin{proposition} \label{PropModelGauss} Let $S={\rm diag}(s_i)$ be an arbitrary diagonal supermatrix of size (possibly infinite) $p'+q'$.  Moreover, set
\beq t_k=\gamma \str S^{-k} = \gamma
\sum_{i=1}^{p'}s_i^{-k}-\gamma\sum_{i=p'+1}^{q'}s_i^{-k} \eeq where $\gamma$ is a complex
parameter.  Then, as long as the $t_k$s are kept arbitrary, one
formally has \beq
 \mathcal{Z}_{(p|q)}(t_1,t_2,\ldots)=(\sdet S)^{\gamma(q-p)}\left\langle \frac{\;\;\prod_{j=1}^{p'}\;\;\sdet(s_j-M)^\gamma}{\prod_{k=p'+1}^{p'+q'}\sdet(s_k-M)^\gamma}\right\rangle_{M\in GU(p|q)}
\eeq
\end{proposition}
 \begin{proof}
  First, suppose that the eigenvalues of $M$ are $(\lambda_1,\ldots,\lambda_p,\lambda_{p+1},\ldots,\lambda_{p+q})$, so the eigenvalues of $IM$ are   $(\lambda_1,\ldots,\lambda_p,\ima\lambda_{p+1},\ldots,\ima\lambda_{p+q})$.
Second,
formally expand $\sum_{k\geq 1} {t_k} \str M^k/k$ in terms of the
eigenvalues and get \beq \frac{\displaystyle \;\prod_{j=1}^{p'}
\;{s_{j}}^{\gamma(p-q)}}{\displaystyle \prod_{j=p'+1}^{p'+q'}
{s_{j}}^{\gamma(p-q)}} \exp\left(- \sum_{k\geq 1} \frac{t_k}{k}\str
M^k\right)=\frac{\;\;\displaystyle\prod_{i=1}^p
\;\;\prod_{j=1}^{p'}\; (s_{j}-\lambda_i)^\gamma \prod_{i=p+1}^{p+q}
\prod_{j=p'+1}^{p'+q'}(s_{j}-\ima\lambda_i)^\gamma}{\displaystyle\prod_{i=1}^p
\prod_{j=p'+1}^{p'+q'}(s_{j}-\lambda_i)^\gamma\prod_{i=p+1}^{p+q}
\;\;\prod_{j=1}^{p'}\;\;(s_{j}-\ima\lambda_i)^\gamma}. \eeq We can
rewrite the last equation  as \beq \exp\left(- \sum_{k\geq 1}
\frac{t_k}{k}\str M^k\right)=(\sdet
S)^{\gamma(q-p)}\frac{\displaystyle\;\prod_{j=1}^{p'}\;\sdet(s_j-IM)^\gamma}{\displaystyle\prod_{j=p'+1}^{p'+q'}\sdet(s_j-IM)^\gamma}
\eeq
 and the proposition follows.
\end{proof}

\section{Loop equations}\label{sectAlg}

In this section, we show that supermatrix models obey the same loop equations (Schwinger-Dyson equations) as usual matrix models.
Since the solution of loop equations of usual matrix models is known order by order in $\hbar$, in terms of invariants of a spectral curve \cite{EOFg}, then this also gives the solution for supermatrix models.
The only novelty, is that supermatrix models can have more general spectral curves than usual matrix models.

\subsection{Schwinger-Dyson equations}

Schwinger-Dyson equations are called loop equations in the context of matrix models.
They merely proceed from the invariance of an integral under changes of variables.
\smallskip

Let us consider the change of variables $f:M\mapsto M'$.  The measure transforms according to the following rule \cite{Berezin}:
\beq
dM'= J(M',M) dM=\sdet\left(\frac{\partial M'}{\partial M}\right)\,dM
\eeq
 $J$ is the Berezinian, which is a generalisation of the Jacobian.  Note that in the last equation,  $\partial M'/\partial M$ is a supermatrix of size $(p+q)^2$ when the size of $M$ is $p+q$.  For instance, if
\beq
M=\begin{pmatrix}
   a&b\\b^*&d
  \end{pmatrix}
\and
M'=\begin{pmatrix}
   \alpha&\beta\\ \beta^*&\delta
  \end{pmatrix}
\eeq
then
\beq
\frac{\partial M'}{\partial M\phantom{'}}=\begin{pmatrix}
                                \frac{\partial \alpha}{\partial a}&  \frac{\partial \delta}{\partial a}&  \frac{\partial \beta}{\partial a} & \frac{\partial \beta^*}{\partial a\phantom{*}}\\
 \frac{\partial \alpha}{\partial d}&  \frac{\partial \delta}{\partial d}&  \frac{\partial \beta}{\partial d} & \frac{\partial \beta^*}{\partial d \phantom{*}}\\
 \frac{\partial \alpha}{\partial b}&  \frac{\partial \delta}{\partial b}&  \frac{\partial \beta}{\partial b} & \frac{\partial \beta^*}{\partial b\phantom{*}}\\
 \frac{\partial \alpha\phantom{*}}{\partial b^*}&  \frac{\partial \delta\phantom{*}}{\partial b^*}&  \frac{\partial \beta\phantom{*}}{\partial b^*} & \frac{\partial \beta^*}{\partial b^*}
                               \end{pmatrix}
\eeq

Now suppose that the transformation $f$ is infinitesimal, that is
\beq\label{EqdeltaM}
f(M)=M'=M+\epsilon g(M) + O(\epsilon^2),
\eeq
for $\epsilon \ll 1$.  By virtue of $\sdet M=\exp (\str \ln M)$, we get
\beq\label{EqdeltaJ}
J(M',M)=1+\epsilon K(M)+\mathcal{O}(\epsilon^2)
\eeq
where
\beq
 K(M)=\str \left(\frac{\partial g(M)}{\partial M}\right).
\eeq

\begin{lemma} \label{lemmaSplit} Let $A$, $B$, and $C$ be supermatrices of the same type than $M$.  Suppose that
\beq
 g(M)=A\frac{1}{x-B M}C.
\eeq
Then we have the ``splitting rule''
\beq\label{EqSplit}
K(M)= \str \left(A\frac{1}{x-B M}B\right) \str \left(\frac{1}{x-B M}C\right).
\eeq
\end{lemma}
\begin{proof}
Firstly, we consider $g_1 (M)=AMB$. According to Eq.\ \eqref{EqdeltaJ},
\beq
K_1 (M)=\str \left(\frac{\partial}{\partial M}AMB\right)=\sum_{i,j}\sigma(i)\sigma(j)\frac{\partial}{\partial M_{ij}}(AMB)_{ij}
\eeq
that is
\beq\label{EqAMB}
K_1 (M)=\sum_{i,j}\sigma(i)\sigma(j) A_{ii}B_{jj}=\str A\str B.
\eeq
Secondly, we choose $g_k(M)=A(BM)^kC$ and successively use $K_1(M)$:
\beq
K_k(M)=\str \left(\frac{\partial}{\partial M}A(B M)^kC\right)=\sum_{\ell=1}^\ell\str \left(A (B M)^{\ell-1} B\right)\str \left((B M)^{k-\ell}C\right).
\eeq
We finally set  $g(M)=A(x-BM)^{-1}C=\sum_{k\geq 0} x^{-k-1} A (B M)^k C$, so that
\begin{equation}
K(M)= \sum_{k\geq 1}\frac{1}{x^{k+1}}K_k(M)
=\sum_{k\geq 0}\sum_{\ell=0}^k\frac{1}{x^{\ell+1}}\str (A(B M)^{\ell}B)\frac{1}{x^{k-\ell+1}}\str ((B M)^{k-\ell}C),
\end{equation}
 which is equivalent to equation we wanted to prove.
\end{proof}

\begin{lemma}Using the above notation, suppose that
\beq
  g(M)=A\str\left(\frac{1}{x-B M}C\right)
\eeq
Then we have the ``merging rule''
\beq\label{EqMerge}
K(M)= \str \left(A\frac{1}{x-B M}C \frac{1}{x-B M}B\right) .
\eeq
\end{lemma}
\begin{proof}
Recall Eqs.\ \eqref{EqdeltaM} and \eqref{EqdeltaJ}.  Direct manipulations lead to
\beq
K_1 (M)=\str (AB)\qquad \text{if}\qquad g_1 (M)= A\str (MB).
\eeq
and, as a consequence,
\beq
K_k (M)=\sum_{\ell=1}^k\str(A(B M)^{k-\ell} C(B M)^{\ell-1}B)\qquad \text{if}\qquad g_k (M)= A\str ((B M)^kC).
 \eeq
Now set $g(M)=A\str\left( ({x-BM})^{-1}C\right)$.  Hence
\begin{equation}
K(M)=\sum_{k\geq 1}\frac{1}{x^{k-1}}K_k (M)=\sum_{k\geq0}\sum_{\ell=0}^k\str\left(\frac{1}{x^{k-\ell+1}}A(B M)^{k-\ell}\frac{1}{x^{ \ell+1}}C(B M)^{\ell}B\right)
\end{equation}
 and the lemma follows.
\end{proof}

Let us find the effect of the transformation  $f:M \mapsto M'=M+\epsilon g(M)$ on matrix integrals.  We define the general expectation value of an analytic  function $G$ as
\beq\label{EqExpGen}
\Big\langle G(M)\Big\rangle_M=\int dM \ee{-\frac{1}{\hbar} \str  V(IM)}G(I M)\Big/ \int dM \ee{- \frac{1}{\hbar} \str  V(IM)}
\eeq
Note that the potential $V$ is  a rational function.     The average of the identity matrix is proportional to
\begin{multline}\label{eqTransInt}
 \int dM' \ee{-\frac{1}{\hbar} \str V(IM') }=
\int dM \ee{ -\frac{1}{\hbar }\str V(IM)}
\Big( 1+\epsilon K(M)\\
-\frac{\epsilon}{\hbar }\sum_{i,j}[Ig(M)]_{ij}\frac{\partial}{\partial (IM)_{ij}} V(IM)+\mathcal{O}(\epsilon^2)\Big)
\end{multline}
However,  the measure is invariant under reparametrisation, which means
\beq
\Big\langle G(M)\Big\rangle_{M }=\Big\langle G(M')\Big\rangle_{M' }.
\eeq
Hence, the following equation must be satisfied:
\beq\label{EqSD1}
 \left\langle \Big( K(I^\dagger M)- \sum_{i,j}\Big(Ig(I^\dagger M)\Big)_{ij}\frac{\partial}{\partial M_{ij}}\str V(M)\Big) \right\rangle_{M }=0.
\eeq
Suppose additionally that
\beq
V(M)=\sum_{k\geq 0}g_k \str M^k-\str MY.
\eeq
From
\beq
\frac{\partial}{\partial M_{ij}}\str M^k=k\sigma(i){M^{k-1}}_{ji}\and \frac{\partial}{\partial M_{ij}}\str MY=\sigma (i)Y_{ji},
\eeq
we get, for any supermatrix $A$,
\beq
\sum_{i,j}A_{ij}\frac{\partial}{\partial M_{ij}}\str V(M)=\str AV'(M).
\eeq
The substitution of the last relation into Eq.\ \eqref{EqSD1} leads to the following  Schwinger-Dyson equation, also known as loop equation.

\begin{lemma}  \label{lemmaSD}
Supermatrix model satisfies the following loop equations:
 \beq
 \left\langle   K(I^\dagger M)\right\rangle_{M }=  \frac{1}{\hbar}\left\langle\str Ig(I^\dagger M)V'(M)  \right\rangle_{M }.
\eeq
\end{lemma}

Notice that since the infinitesimal Berezinian $K$ satisfies the same split and merge rule as usual matrix models, we already have that supermatrix models satisfy the same loop equations as usual matrix models.

\subsection{Loop equations for the supermatrix model in an external field}

Here, we are interested in $V(M)=\str v(M)-MY$.

Moreover, in order to close the set of loop equations, it was found in \cite{EOFg} that one should consider loop equations for the following expectation values, or more precisely their connected parts or (joint) cumulants  $\langle G\rangle^c$. \footnote{Cumulants can be defined via
$
\langle G_1(M)\ldots G_n(M)\rangle=\sum_{\pi\{1,\ldots,n\}}\prod_{J\in \pi}\langle G_{j_1}G_{j_2}\ldots\rangle^c,
$
 where the sum is over all partitions $\pi$ of the set $\{1,\ldots,n\}$ while $J=\{j_1,j_2,\ldots\}$ is an element of $\pi$.  For instance
$
\langle AB\rangle=\langle AB \rangle^c+ \langle A \rangle^c\langle   B\rangle^c
$
and
$
\langle ABC\rangle=\langle ABC\rangle^c+ \langle AB  \rangle^c\langle   C\rangle^c+\langle AC  \rangle^c\langle   B \rangle^c+\langle BC  \rangle^c\langle   A\rangle^c+\langle A\rangle^c\langle B \rangle^c\langle   C\rangle^c.
$
}
Specifically, we consider
\beq\label{defwbar}
\bar{w}(z_1,\ldots,z_k)=\left\langle\prod_{i=1}^k\str\frac{1}{z_i-M}\right\rangle^c_M
\eeq
\beq\label{defubar}
\bar{u}(x,y;z_1,\ldots,z_k)=\left\langle \str\frac{1}{x-M} \frac{\mu(y)}{y-Y}\prod_{i=1}^k\str\frac{1}{z_i-M}\right\rangle^c_M
\eeq
\beq\label{defubar}
\bar{p}(x,y;z_1,\ldots,z_k)=\left\langle \str\frac{v'(x)-v'(M)}{x-M} \frac{\mu(y)}{y-Y}\prod_{i=1}^k\str\frac{1}{z_i-M}\right\rangle^c_M
\eeq
where $\mu(y)=\mu(y;Y)$ is the minimal polynomial of $Y$, i.e. the polynomial $\prod_j(y-y_j)$, where the product is over the disctinct eigenvalues of $Y$.

Notice that $\bar{u}$ and $\bar{p}$ are polynomials in the variable $y$, and, if $v'(x)$ is a rational fraction of $x$, then  $\bar{p}(x,y;z_1,\dots,z_k)$ is also a rational fraction of $x$ with the same poles, and with degree one less than $v'(x)$.

\smallskip
In terms of those expectation values, we have the following loop equations:

\begin{proposition}Let $Y$ and $M$ be Hermitian supermatrices  of size $p+q$.  In   Eq.\ \eqref{EqExpGen}, set
\beq
 V(M)=v(M)-MY.
\eeq
 Then, for every set of variables $J=\{z_1,\dots,z_k\}$, we have the loop equation:
 \begin{multline}\label{loopeq1}
\bar{u}(x,y;J\cup\{x\}) + \sum_{I\subset J} \bar{u}(x,y;I)\bar{w}(x,J\setminus I) + \sum_{j=1}^k {\partial \over \partial z_j}\,
{\bar{w}(x,J \setminus \{z_j\})-\bar{w}(z_j,J\setminus \{z_j\})\over x-z_j}  \\
=  {1\over \hbar}\Big[ (v'(x)-y)\bar{u}(x,y;J) -\bar{p}(x,y;J) +\mu(y) \bar{w}(x,J) \Big]
 \end{multline}
which is the same loop equation as the usual 1-matrix model in an external field.

\end{proposition}

\begin{proof} We   choose
\beq
g(M)=I^\dagger \frac{1}{x-IM}\frac{\mu(y)}{y-Y}\, \prod_{j=1}^k \str\frac{1}{z_j-M}
\eeq
and apply the splitting rule of Lemma \ref{lemmaSplit}  to the Schwinger-Dyson equation of Lemma \ref{lemmaSD}.  A few manipulations complete the proof.
\end{proof}

Since loop equations arise from local changes of variables, it is clear that
loop equations may have many solutions, in fact as many as possible open integration domains in which matrices are integrated.
Here, the integration domain is specified by defining our matrix integral as formal series, such that each term in the series is a polynomial moment of a Gaussian integral, like in Eq.\ \eqref{EqPartFormal}.

Moreover we have seen from theorem \ref{PropCentral}, that our formal supermatrix integrals have a topological expansion
\beq
\bar{w}(x_1,\dots,x_n) = \sum_{g=0}^\infty \hbar^{2g-2+n}\,\bar{w}^{(g)}(x_1,\dots,x_n)
\eeq

It was found in \cite{EOFg}, that there is a unique formal power series solution of loop equations having such a topological expansion, and that unique solution was computed in terms of the "symplectic invariants" of a spectral curve.
We explain below how to find the spectral curve.

\subsection{Spectral curve}
\label{SectionSpectral}

The spectral curve can be found from the loop equation with $J=\emptyset$; that is,
\beq
\hbar\, \bar{u}(x,y;x)
=   (v'(x)-y-\hbar \bar{w}(x))\,\bar{u}(x,y) -\bar{p}(x,y) +\mu(y) \bar{w}(x)
 \eeq
 and if we expand it into powers of $\hbar$, to leading order we have:
\beq
0 =   (v'(x)-y-\bar{w}^{(0)}(x))\,\bar{u}^{(0)}(x,y) -\bar{p}^{(0)}(x,y) +\mu(y) \bar{w}^{(0)}(x)
\eeq
We see this functional equation is greatly simplified if the second term disappears, i.e., if we choose $y=v'(x)-\bar{w}^{(0)}(x)$.  Moreover, in order to obtain algebraic relations, we introduce  $D(x)$, the denominator of $v'(x)$.

\begin{corollary}\label{coroSpectral}   Define two scalar functions of one and two complex variables respectively:
\beq
\mathcal{Y}(x)=v'(x)-\bar{w}^{(0)}(x)
\eeq
and $E_\mathrm{ext}$ which is a polynomial of its two variables such that
\beq\label{defspcurveEext}
D(x)^{-1}\,E_\mathrm{ext}(x,y)=\left(v'(x)-y\right)\mu(y)-\bar p^{(0)}(x,y).
\eeq
 Then, the following algebraic equation holds:
\beq\label{Extxy0}
E_\mathrm{ext}(x,\mathcal{Y}(x))=0.
\eeq
\end{corollary}
The algebraic plane curve defined by this equation is called the spectral curve.

As we said, the unique formal series solution of those loop equations of the form \eqref{loopeq1}, having a topological expansion in $\hbar^2$,  was computed in  \cite{EOFg} in terms of the "symplectic invariants" of the spectral curve:
\beq
{\mathcal E}(x,y)=E_\mathrm{ext}(x,y) = 0.
\eeq
It fact, to any spectral curve ${\mathcal E}$, one can associate an infinite sequence of numbers $\mathcal F^{(g)}({\mathcal E})$, $g=0,1,2,\dots$.
The $\mathcal F^{(g)}$'s are computed in terms of  the spectral curve, and are residues of birational expressions of $x$ and $y$. We refer the reader to \cite{EOFg} for detailed computations of the $\mathcal F^{(g)}$'s.
The result is that:
\beq
 \ln{Z} = \sum_{g=0}^\infty\, \hbar^{2g-2}\, \mathcal F^{(g)}({\mathcal E})
\eeq

The functions  $\mathcal F^{(g)}$ are called the symplectic invariants of the spectral curve because, as exposed in the following theorem, they remain unchanged if one changes the spectral curve without changing the symplectic form $dx\wedge dy$.

\begin{theorem}\cite{EOFg} \label{TeoSymplectic} For all $g\geq 2$, the free
energy $\mathcal{F}^{(g)}(\mathcal{E})$ is invariant under the
following  transformations of the algebraic equation
$\mathcal{E}(x,y)=0$:
\beq\begin{array}{lll}
(1) & x\mapsto x,  &y\mapsto y+R(x)  \\
(2)& x\mapsto c x, & y\mapsto y/c  \\
(3)& x\mapsto -x, & y\mapsto y  \\
(4)& x\mapsto y , & y\mapsto x
\end{array}\eeq
where $R$ is a rational function and $c$ is a complex number.
$\mathcal{F}^{(0)}$ is invariant under the four transformations if
$R$ is replaced by a polynomial $P$.  $\mathcal{F}^{(1)}$ is
invariant under the four transformations up to an additive factor of
$\ima \pi/12$.  These transformations preserve, up to a sign, the
symplectic form $dx\wedge dy$.\end{theorem}

We stress that previous theorem implies that for all $g\geq0$, $\mathcal{F}^{(g)}({\mathcal E})$ is invariant under the exchange of $x$ and $y$ in the $\mathcal{E}(x,y)=0$.  This invariance  will be exploited for proving the duality for the Gaussian model.

\section{Gaussian model with sources and external fields}

\subsection{Spectral curve}
We now focus on the study of the partition function given in Eq.\ \eqref{eqPartFS}.  For this, we choose
\beq
V(M)=v(M)-MY, \qquad  v(M)=\frac{1}{2}M^2-\hbar\ln\prod_{i=1}^{m+n}\left(x_i-M\right)^{\sigma_{m,n}(i)}
\eeq
 where $M$ is a Hermitian supermatrix of size $p+q$.  We  suppose that the variables $x_i$'s are the eigenvalues of a Hermitian matrix of size $m+n$.  Amongst these eigenvalues  only $m'+n'$, say,  are distinct.  This means that
\beq
\sdet (z-X)=\prod_{i=1}^{m+n}(z-x_i)^{\sigma_{m,n}(i)}=\prod_{i=1}^{m'+n'}(z-x_i)^{a_i}
\eeq
where the signed multiplicities $a_i$ are such that
\beq
m=\sum_{i=1}^{m'}a_i \and n=-\sum_{i=m'+1}^{m'+n'}a_i.
\eeq
Similarly,
\beq
\sdet (z-Y)=\prod_{i=1}^{p+q}(z-y_i)^{\sigma_{p,q}(i)}=\prod_{i=1}^{p'+q'}(z-y_i)^{b_i}.
\eeq
Thus, the minimal polynomial of $Y$ and the denominator of $v'$ are respectively given by
\beq
\mu(y)=\mu(y;Y)=\prod_{i=1}^{p'+q'}(y-y_i) \and
D(x)=\mu(x;X)=\prod_{j=1}^{m'+n'}(x-x_i)
\eeq
The spectral curve given in Corollary \ref{coroSpectral} , now becomes
\begin{multline}\label{spcurvegenGaus}
\mathcal{E}(x,y)=E_\mathrm{ext}(x,y)=\prod_{i} (x-x_i)\prod_{j} (y-y_j)\left(x-\hbar\sum_{i}\frac{a_i}{x-x_i}\right.\\
\left.-y-\hbar\sum_{j} \frac{b_j}{y-y_j}+\hbar^2\sum_{i,j} \frac{a_i}{x-x_i}\frac{b_j}{y-y_j}\left\langle \left( \frac{1}{x_j-M}\right)_{jj}\right\rangle\right)
\end{multline}
where it is assumed that $i\in\{1,\ldots,m'+n'\}$ and $j\in\{1,\ldots,p'+q'\}$.

\smallskip

\begin{remark}
This is more or less the same spectral curve as in the usual hermitian matrix model in an external field, with one notable difference.
In the usual matrix model, the coefficients $b_j$ are necessarily positive integers, and here, supermatrix models allow more general spectral curves, where $b_j$'s can also be negative integers.
\end{remark}

\subsection{Geometry of the spectral curve}

The geometry of spectral curves of the type \ref{spcurvegenGaus} has been studied many times, and is a standard exercise of Riemann geometry \cite{Farkas, Fay}. Here, we only briefly summarize the main points, and we follow the same lines as \cite{EOFg}.

The spectral curve $\mathcal{E}(x,y)=0$ defines a unique compact Riemann surface.
Generically, this algebraic Riemann surface has genus at most $(m'+n')(p'+q')-1$.

The genus can be lower than that. In particular, if our supermatrix model is defined as perturbation of a Gaussian integral near $M=0$, (as we did in Proposition \ref{PropModelGauss}), then the spectral curve must have genus $0$, i.e. it must be a rational spectral curve.

But in general, the genus can be anything between $0$ and $(m'+n')(p'+q')-1$.

Notice that in Eq.\ \eqref{spcurvegenGaus}, the coefficients $\left\langle \left( \frac{1}{x_j-M}\right)_{jj}\right\rangle$ have not been determined.
Those coefficients cannot be determined by loop equations, they are related to the integration domain for the supermatrix integral.
Those coefficients are in 1-1 correspondence with the  "filling fractions":
\beq\label{ffcondition}
\epsilon_i = {1\over 2i\pi}\,\oint_{{\mathcal A}_i} ydx
\eeq
where ${\mathcal A}_i, i=1,\dots, {\rm genus}$ is a set of independent non-contractible cycles on the Riemann surface.

A choice of integration domain, is equivalent to a choice of those coefficients, and thus is equivalent to a choice of filling fractions.

\smallskip

The spectral curve is then described by this Riemann surface, and by two meromorphic functions $x(z)$ and $y(z)$ defined on that surface.
Those two functions can be completely described by their poles and by their cycle integrals \eqref{ffcondition}.

\subsection{Poles}

The spectral curve can be studied by determining the singularity structure of
the algebraic equation.   More precisely, let two complex function
$x$ and $y$.  Let $z$ belong to a the Riemann surface associated to
the algebraic equation if $\mathcal{E}(x(z), y(z))=0$ for all z on
the surface.

In our case, it is clear that the algebraic equation \eqref{spcurvegenGaus} has only three
types of singularity, all being poles.
First, it is obvious that the polynomial  $\mathcal{E}(x,y)$ diverges when both $x$ and $y$ become
infinite. The equation $\mathcal{E}(x,y)=0$ gives in that limit:
\beq
y \sim x - {\hbar\,(\sum_i a_i+\sum_j b_j)\over x} + O(1/x^2)=x - {\hbar\,(m-n+p-q)\over x} + O(1/x^2)
\eeq
Second,  we see that $\mathcal{E}$ also diverges when $x$ tends to
$x_i$.  We call $\xi_i$ the point on the surface such that
$x(\xi_i)=x_i$ and, in order to comply with the algebraic
equation $\mathcal{E}(x,y)=0$, we must have in that limit:
\beq
y \sim -\,{\hbar\, a_i\over x-x_i} + O(1)
\eeq
Third, the polynomial goes to infinity if $y$ approaches $y_i$, which
correspond to a diverging $x$. We thus set $y(\eta_i)=y_i$ and we have in this limit:
\beq
x\sim {\hbar\, b_i\over y-y_i} + O(1)
\eeq

The above equations yield the following characterization of $x$ and $y$ and, as a consequence, of the spectral curve.

\begin{lemma}\label{LemmaCurve}
The meromorphic functions $x$ and $y$, have simple poles at the points:
\begin{equation}\label{polestructurexy}
 \begin{array}{lcll}z=\infty &:& x=\infty,&
y=\infty\\
z=\xi_i & : & x=x_i,& y=\infty\\
z=\eta_j & : & x=\infty, &y=y_j.
\end{array}
\end{equation}
The spectral curve ${\mathcal E}_{{(m|n),(p|q)},(X,Y,\epsilon)}$ is complectly characterized by the residues
\begin{equation}\label{poleresxdyetai}
 \begin{split}
\Res_{z\to\infty} ydx &= \phantom{-}\hbar(m-n+p-q) \\
\Res_{z\to\infty} x dy &= - \hbar(m-n+p-q)
\\
\Res_{z\to \xi_i} ydx &= -\hbar \, a_i
\\
 \Res_{z\to \eta_i} xdy &= \phantom{-}\hbar\, b_i
\end{split}
\end{equation}
and the filling fractions
\beq
\epsilon_i = {1\over 2i\pi}\,\oint_{{\mathcal A}_i} ydx \label{ffcondition1}
\eeq
where we recall that a choice of integration domain, is equivalent to a choice of filling fractions.
\end{lemma}

\section{Duality}

We have determined the spectral curve, which allows in principle  to construct the free energy via the following expansion:
Then, when we have determined the spectral curve, we have:
\beq\label{EqDefZFcurve}
\ln{(Z_{(m|n),(p|q)}(X,Y,\epsilon))} = \sum_{g=0}^\infty \hbar^{2g-2}\,\, \mathcal{F}_g(\mathcal{E}_{{(m|n),(p|q)},(X,Y,\epsilon)}),
\eeq
where $\mathcal{F}_g(\mathcal{E}_{{(m|n),(p|q)},(X,Y,\epsilon)})$ denotes the symplectic invariant $\mathcal{F}_g(\mathcal{E})$ defined
in \cite{EOFg} for the particular spectral curve $\mathcal{E}=\mathcal{E}_{{(m|n),(p|q)},(X,Y,\epsilon)}$.

\begin{theorem} \label{TeoDualityFilling} The following duality holds:
\beq
Z_{(m|n),(p|q)}(X,Y,\epsilon) = Z_{(p|q),(m|n)}(Y,X,-\epsilon)
\eeq
\end{theorem}
\begin{proof}
Looking at the spectral curve characterized by data given in Lemma \ref{LemmaCurve}, we see that $\mathcal{E}_{{(m|n),(p|q)},(X,Y,\epsilon)}$ and $\mathcal{E}_{{(p|q),(m|n)},(Y,X,-\epsilon)}$ are just obtained from one another by the exchange of $x$ and $y$.  We now from Theorem \ref{TeoSymplectic} that the $\mathcal{F}_g$ are invariant under such a transformation, so that the result follows from Eq.\ \eqref{EqDefZFcurve}.
\end{proof}
The change $\epsilon\to -\epsilon$ is just the change of orientation for the integration contours used to define the integrals.

 In case we are studying the perturbative supermatrix integral (small deformation of the Gaussian integral), we need a genus zero spectral curve.  Indeed, in such instance, the parameter $z$ belongs to the complex plane, and $x(z)$ and $y(z)$ are rational fractions of $z$ whose poles are fixed by the analysis performed in the previous section, i.e. by equations \eqref{polestructurexy} to \eqref{poleresxdyetai}.  We therefore find:
\begin{lemma}\label{LemaRational} For the perturbative matrix integral, the rational spectral curve can be parametrized as follows:
\beq
{\mathcal E}_{{(m|n),(p|q)},(X,Y)}
=\left\{
\begin{array}{ll}
x(z)&=z+\hbar\sum_i\frac{b_i}{y'(\eta_i)(z-\eta_i)}\\
y(z)&=z-\hbar\sum_i\frac{a_i}{x'(\xi_i)(z-\xi_i)}.
\end{array}
\right.
\eeq
where the complex numbers $\xi_i$, $\eta_i$ are obtained by solving $x(\xi_i)=x_i$ and $y(\eta_i)=y_i$.
\end{lemma}

This means that in the rational case, there is no filling fraction.  The duality exposed in the introduction is thus a corollary
 of Theorem \ref{TeoDualityFilling} and Lema \ref{LemaRational}.

We end this section with a simple application of the duality:
we show that the expectation of a product of characteristic polynomials in a Gaussian supermatrix model formally tends to a
generalization of the Kontsevich model (matrix Airy function) when $1/(p-q)\to 0$.  We first make use of Eq.\ \eqref{eqGaussIntA},
$$Z_{(m|n),(p|q)}(X,0)=Z_{(p|q),(m|n)}(0,X),\quad\text{and}\quad\sum_{i=1}^{p+q}\sigma(i)=p-q,$$
which yields
\beq \label{EqK}
 Z_{(m|n),(p|q)}(\ima A,0)= \frac{\ee{\frac{1}{2\hbar}\str A^2}}{z_{m,n}(\hbar)}\int_{H(m|n)} dN \ee{-\frac{1}{2\hbar}\str(IN)^2+ \frac{\ima}{\hbar}\str INA} \sdet(-IN)^{p-q},
\eeq
In the last equation,  $A=\mathrm{diag}(a_1,\ldots,a_{m+n})$ and $z_{m,n}(\hbar)$ stands for the normalization coefficient defined in Eq.\ \eqref{EqCoeffz}.
Now, let  $W$ and $X$  be a Hermitian supermatrix of size $m+n$  and a diagonal matrix $X$ with $m+n$ complex entries, respectively.
We also set
\beq \hbar=(p-q)^{-1}, \quad IN=\ima +\hbar^{1/3}W, \quad \text{and}\quad A=2+\hbar^{2/3}X .\eeq
Finally, we formally expand the integrand in Eq.\ \eqref{EqK} in powers of $\hbar=(p-q)^{-1}$ and conclude that, as $\hbar\to0$,
\beq
\frac{\ee{-\frac{1}{\hbar^{1/3}}\str X}}{\bar z_{m,n}(\hbar)}\,Z_{(m|n),(p|q)}(2\ima+\ima \hbar^{1/3} X,0)= \mathrm{Ai}(X)+O(\hbar^{1/3}),
\eeq
where \beq \mathrm{Ai}(X)=\int_{H(m|n)}dW\ee{\frac{\ima}{3}\str W^3+\ima\str WX} \eeq
and \beq \bar z_{m,n}(\hbar)=\frac{\ee{\frac{m-n}{2\hbar}}\ima^{(m-n)(\hbar^{-1}-n)}}{\hbar^{(m-n)^2/3}}z_{m,n}(\hbar).\eeq

\section{Conclusion}

We have proved that gaussian supermatrix integrals with external fields and sources, respectively denoted by $Y=\mathrm{diag}(y_1,\ldots y_{p+q})$
 and $X=\mathrm{diag}(x_1,\ldots,x_{m+n})$ for some complex numbers $y_i$ and $z_i$, satisfy a duality formula which extends that of \cite{Brezin,Des}.
Usual hermitian matrices with sources at both numerators and denominators could not have this duality,
 because we see that numerators and denominators are transformed, under this duality, into variables of different signs,
 which can be only obtained with supermatrices. In some sense supermatrices allow eigenvalues with negative multiplicities.
We conjecture that the supermatrix duality proved in the article extends to the case where the external sources and fields, $X$ and $Y$, are arbitrary supermatrices.



\section*{Acknowledgments}
The authors would like to thank N. Orantin
for useful and fruitful discussions on this subject.
The work of B.E. is partly supported by the Enigma European network MRT-CT-2004-5652, by the ANR project G\'eom\'etrie et int\'egrabilit\'e en physique math\'ematique ANR-05-BLAN-0029-01,
by the European Science Foundation through the Misgam program,
by the Quebec government with the FQRNT.  The work of P.D. is supported by the  Fondo Nacional de Desarrollo Cient\'ifico y Tecnol\'ogico (FONDECYT) Grant \#1090034.  The project was initiated while P.D. was a postdoctoral fellow at the IPhT, CEA-Saclay; P.D. wishes to thank the members of the IPhT for their kind generosity and hospitality.

\appendix

\section{Grassmann variables}

The Grassmann algebra of order $n$, denoted $\Gamma=\Gamma(n;\mathbb{C})$, is the algebra over $\mathbb{C}$ generated by $1$ and $n$ quantities $\theta_i$, called Grassmann variables, which satisfy
\begin{equation}
 \theta_i\theta_j=-\theta_j\theta_i.
\end{equation}
Note in particular that $\theta_i^2=0$.  Let $c$ be a complex number.  We adopt the following convention for the complex conjugation:
\begin{equation}
 (c\theta_i)^*=c^*\theta_i^*,\qquad (\theta_i\theta_j)^*=\theta_j^*\theta_i^*,\qquad \theta_i^{**}=\theta_i.
\end{equation}

$\Gamma$ has $2^n$ generators: Every element $x$ can be decomposed as
\begin{equation}
 x=x^{(0)}+\sum_{i_1} x^{(i_1)}\theta_{i_1}+\sum_{i_1<i_2}x^{(i_1,i_2)}\theta_{i_1}\theta_{i_2}+\ldots+x^{(1,\ldots,n)}\theta_{1}\cdots\theta_n
\end{equation}
with $x^{(i_1,i_2,\ldots)}\in\mathbb{C}$.  We say that $x$ is even (or bosonic) if it contains only monomials with an even number of $\theta_i$; $x$ is odd (fermionic) if contains only monomials with an odd number of $\theta_i$.  Elements $x_i$ a the Grassmann algebra satisfy
\begin{equation}
 x_i x_j=(-1)^{(\deg x_i \cdot \deg x_j) }x_j x_i
\end{equation}
where $\deg x=0$ if $x$ is even while $\deg x=1$ if $x$ is odd.  Note that we also write
\beq
 \deg x = \epsilon (x)\and (-1)^{\deg x}=\sigma (x).
\eeq

The (left) derivative and the integration with respect to the Grassmann variables are respectively defined by
\begin{equation}
\frac{\partial}{\partial \theta_i}( \theta_{j_1} \theta_{j_2}\cdots \theta_{j_k})=\delta_{i,j_1}(\theta_{j_2}\cdots\theta_{j_k})-\delta_{i,j_2}(\theta_{j_1}\theta_{j_3}\cdots\theta_{j_k})+(-1)^k\delta_{i,j_k}(\theta_{j_1}\cdots\theta_{j_{k-1}})
\end{equation}
and
\begin{equation}
 \int d\theta_i\,\theta_j=  -\int\theta_j\,d\theta_i =\delta_{i,j}.
\end{equation}
With the latter definition, one can obtain the following integral representation for the determinant of a $N\times N$ matrix $X$:
\begin{equation}
 \det X=\int d\theta^\dagger d\theta \ee{-\theta^\dagger X\theta}
\end{equation}
where $\theta$  (resp.\ $\theta^\dagger$) is interpreted as a column (resp.\  row)  vector with $N$ fermionic components.   The bosonic counterpart of the latter formula is the usual Gaussian integral involving complex variables $z_i$; that is, if $X$ is Hermitian,
\begin{equation}
 \frac{1}{\det X}=\frac{1}{(2\pi\ima)^N}\int  d{z}^\dagger dz \ee{ -z^\dagger X z}.
\end{equation}

\section{Supermatrices}

Let $X$ denote an even (bosonic) square supermatrix of size $(p+q)$.  It can be written as
\begin{equation}
X=\Big( X_{ij} \Big)_{1\leq i,j\leq p+q}=\begin{pmatrix}A&B\\C&D\end{pmatrix}
\end{equation}
 where $A$ and $D$  are respectively $p\times p$ and $q\times q$ matrices with even Grassmann elements, while  $B$ and $C$  are respectively $p\times q$ and $q\times p$ matrices with Grassmann odd elements.   The following notation is useful:
\beq
(-1)^{\deg X_{ij}}=\sigma(i)\sigma(j)
\eeq
where
\beq
\sigma(i)=(-1)^{\epsilon(i)}=
\begin{cases}
  +1,& i=1,\ldots p,\\
  -1, & i=p+1,\ldots p+ q.
\end{cases}
\eeq

 The supertrace is given by
\begin{equation}
 \str X=\sum_{i=1}^{p+q}\sigma(i)X_{ii}=\tr A-\tr D
\end{equation}
and satisfies
\begin{equation}
 \str (X+Y)=\str X+\str Y,\qquad \str(XY)=\str (YX).
\end{equation}
  The superdeterminant, which exists only if $A$ as well as $D$ are invertible, is given by
\begin{equation}
 \sdet X=\frac{\det(A-BD^{-1}C)}{\det (D)}=\frac{\det(A)}{\det (D-CA^{-1}B)}.
\end{equation}
 One can show that
\begin{equation}
 \sdet(XY)=\sdet X\sdet Y\qquad\text{and}\qquad\det(\exp X)=\exp (\str X).
\end{equation}

We use the following definitions for the transpose and the adjoint of  a supermatrix:
\begin{equation}
 X^{\trans}=\begin{pmatrix}A^\trans &C^\trans \\ B^\trans & D^\trans \end{pmatrix}\qquad\mbox{and}\qquad  X^{\dagger}=\begin{pmatrix}A^\dagger &C^\dagger \\ B^\dagger & D^\dagger \end{pmatrix}
\end{equation}
where $A^\dagger$ means $(A^\trans)^*$.  On easily shows that
\begin{equation}
(X^\dagger)^{\dagger}=X,\qquad (XY)^{\dagger}=Y^\dagger X^\dagger,\qquad \sdet X^\dagger=(\sdet X)^*.
\end{equation}
However, $(XY)^\trans\neq Y^\trans X^\trans$ and $\sdet X^\trans\neq \sdet X$ in general.  \footnote{Other non equivalent definitions for the transpose and the adjoint are possible. They lead, for instance, to a distinct unitary supergroup, namely $sU(p|q)$.  \cite{Berezin, Frappat}.}  A supermatrix $X$ is Hermitian if $X^\dagger=X$; it is unitary if $X^\dagger=X^{-1}$.  The set of all invertible even supermatrices of size $p+q$, whose elements belong to a Grassman algebra over $\mathbb{C}$, form the general linear super group $GL(p|q)$.   All unitary supermatrices form the superunitary group $U(p|q)$.

Let $\xi^\dagger=(z_1^*,\ldots,z_p^*, \theta_1^*,\ldots,\theta_q^*)$ be the adjoint of the supervector $\xi$.  Then,  the superdeterminant of a Hermitian supermatrix $X$ has the following integral representation:
\begin{equation}
 \frac{1}{\sdet X}=\frac{1}{(2\pi \ima)^p}\int d\xi^{\dagger} d\xi \ee{-\xi^\dagger X\xi}.
\end{equation}

The measure for Hermitian supermatrices is given by the  product $\prod_{i,j} dX_{i,j}$  which is equal, up to a multiplicative constant, to the product of real independent differential elements.   Explicitly, if $X\in H(p|q)$,
\beq
dX=\prod_{1\leq i\leq p+q}dX_{ii} \prod_{1\leq i<j\leq p}dX_{ij}dX^*_{ij}\prod_{\substack{1\leq i\leq p\\
p+1\leq j\leq p+q}}dX_{ij}dX^*_{ij}.
\eeq
For instance, by using the bosonic and fermionic Gaussian integrals,
\beq
\int_\mathbb{R}dx \ee{-x^2+xy}=\sqrt{ {\pi} }\,\ee{\frac{1}{4}y^2}\and    \int d\theta d\theta^* \ee{\theta^*\theta+\theta^*\eta+\eta^*\theta}=\ee{\eta\eta^*},
\eeq
where $\deg \eta=\deg \eta^*=1$, one readily shows that the following formulas hold for any (not necessarily Hermitian) supermatrix  $Y$ of size $p+q$:
\beq\label{eqGaussIntA}
\int dX \ee{-\frac{1}{2\hbar}\str  (I X)^2}\ee{\str IXY }
=z_{p,q}(\hbar)\ee{\frac{\hbar}{2}\str Y^2},
\eeq
where
\beq \label{EqCoeffz}z_{p,q}(\hbar)=2^{(p+q)/2}\ima^{pq}\pi^{(p^2+q^2)/2}\hbar^{(p-q)^2/2}, \eeq
so that
\beq\label{eqGaussInt}
\left\langle \ee{\str XY}\right\rangle_{X\in GU(p|q)} =
\frac{\int dX \ee{-\frac{1}{2\hbar}\str  (I X)^2}\ee{\str IXY }}{ \int dX \ee{-\frac{1}{2\hbar}\str  (I X)^2}}=\ee{\frac{\hbar}{2}\str Y^2}.\eeq



\end{document}